\documentclass[11pt]{amsart}

\usepackage[margin=1in]{geometry}
\usepackage{amsmath,amssymb,amsthm,mathtools}
\usepackage{microtype}
\usepackage{xcolor}
\usepackage{hyperref}
\usepackage{aliascnt}
\usepackage[nameinlink,noabbrev]{cleveref}

\hypersetup{
  colorlinks=true,
  hypertexnames=false,
  linkcolor=blue!55!black,
  citecolor=blue!55!black,
  urlcolor=blue!55!black,
  pdftitle={Diffuse Gaussian Truncation for Deterministic Approximate Counting},
  pdfauthor={Zihong Yi}
}
\newtheorem{theorem}{Theorem}[section]
\newaliascnt{lemma}{theorem}
\newtheorem{lemma}[lemma]{Lemma}
\aliascntresetthe{lemma}
\newaliascnt{proposition}{theorem}
\newtheorem{proposition}[proposition]{Proposition}
\aliascntresetthe{proposition}
\newaliascnt{corollary}{theorem}
\newtheorem{corollary}[corollary]{Corollary}
\aliascntresetthe{corollary}
\newtheorem*{theorem*}{Theorem}
\theoremstyle{definition}
\newaliascnt{definition}{theorem}
\newtheorem{definition}[definition]{Definition}
\aliascntresetthe{definition}
\theoremstyle{remark}
\newaliascnt{remark}{theorem}
\newtheorem{remark}[remark]{Remark}
\aliascntresetthe{remark}

\crefname{theorem}{Theorem}{Theorems}
\crefname{lemma}{Lemma}{Lemmas}
\crefname{proposition}{Proposition}{Propositions}
\crefname{corollary}{Corollary}{Corollaries}
\crefname{definition}{Definition}{Definitions}
\crefname{remark}{Remark}{Remarks}
\crefname{section}{Section}{Sections}
\crefname{appendix}{Appendix}{Appendices}
\crefname{equation}{}{}

\newcommand{\PM}{\operatorname{PM}}
\newcommand{\haf}{\operatorname{haf}}
\newcommand{\tr}{\operatorname{tr}}
\newcommand{\E}{\mathbb E}
\newcommand{\R}{\mathbb R}
\newcommand{\C}{\mathbb C}
\newcommand{\1}{\mathbf 1}
\newcommand{\ind}[1]{\mathbf 1_{\{#1\}}}
\newcommand{\eps}{\varepsilon}
\newcommand{\poly}{\operatorname{poly}}
\newcommand{\overviewstage}[2]{%
  \par\addvspace{\medskipamount}%
  \noindent\emph{#1.\ #2.}\enspace\ignorespaces}
\begin{document}


\title[Diffuse Gaussian truncation]
{Diffuse Gaussian Truncation for Deterministic Approximate Counting}
\author{Zihong Yi}
\thanks{Carnegie Mellon University. E-mail:
\href{mailto:zihongy@andrew.cmu.edu}{zihongy@andrew.cmu.edu}.}
\date{}

\begin{abstract}
We give deterministic fully polynomial-time approximation schemes for two
dense counting problems on which the known deterministic algorithms, based
on zero-free interpolation, run in quasipolynomial time.
For fixed $0<\gamma<1/2$ and $0<\theta\le1$, the first approximates
$\haf(A)$ for a symmetric matrix $A$ of even order $n$ when its support graph
$G$ has minimum degree at least $(1/2+\gamma)n$ and its nonzero entries lie
in $[\theta,1]$.  It also approximates permanents under the analogous
bipartite condition, including all full-support matrices with entries in
$[\theta,1]$.  For fixed $\beta>0$ and $0<\kappa\le1$, the second
approximates the zero-field Ising partition function $Z(J)$ for
zero-diagonal real symmetric matrices $J$ satisfying
$\max_{i,j}|J_{ij}|\le\beta/n$ and
$\lambda_{\max}(J)\le1-\kappa$.  This spectral condition is one-sided, with
no separate lower-eigenvalue hypothesis.  We further prove
$\log\haf(A)=h_A(G)-n/2+O_{\gamma,\theta}(1)$ and
$Z(J)=2^n\det(I-J)^{-1/2}(1+O_{\beta,\kappa}(1/n))$.
Here $h_A(G)$ is the maximum weighted fractional-matching entropy.  For
unweighted graphs, the first formula improves the Cuckler--Kahn error from
$o(n)$ to $O_\gamma(1)$ on the fixed-margin class.  It also extends their
formula to weights in $[\theta,1]$.

Both algorithms follow from a common Gaussian truncation principle.  Each
problem reduces to integrals of products of a fixed entire function over
Gaussian coordinates with a possibly indefinite moment matrix whose entries
are $O(1/n)$.  After linear cancellation and exact quadratic resummation,
the coordinate remainder vanishes to order at least three.  Complex dilation
controls small supports.  For large supports, we bound the recombined tail by
a large-deviation estimate whose rate beats the entropy of the subsets.
Together these give truncation error $(CR/n)^{R/2}+e^{-cn}$.  This
faster-than-geometric decay permits a cutoff satisfying
$R\log(en/R)=O(\log n+\log(1/\eps))$, which makes the enumeration
polynomial.  The retained terms are evaluated by monomer--dimer recurrences
and one-dimensional quadrature for matchings, and by small spin sums for
Ising.
\end{abstract}

\maketitle

\section{Introduction}
\label{sec:introduction}

For the permanent and the Ising partition function on the dense inputs
considered below, known deterministic approximation algorithms remain
quasipolynomial.  The Taylor expansion behind zero-free interpolation
converges geometrically, so attaining accuracy $\eps$ requires
$\Theta(\log(1/\eps))$ layers, while
evaluating through layer $R$ costs $n^{\Theta(R)}$.  We take a different
route.  After passing to a Gaussian product with a diffuse moment matrix,
whose entries are $O(1/n)$, we exactly resum its quadratic part.  The
contribution supported on $R$ coordinates is then suppressed by
$(R/n)^{\alpha R}$ for some fixed $\alpha>0$.  This faster decay permits a
cutoff with only polynomially many retained terms.

The resulting principle gives deterministic FPTASs for dense weighted
perfect matchings and permanents, and for zero-field Ising models with diffuse
couplings.  The two applications arrive at the Gaussian product in different
ways.
For matchings, maximum-entropy scaling removes the main exponential
contribution, and an inverse-Gamma identity produces a one-parameter family
of Gaussian products.  For Ising, the Hubbard--Stratonovich identity gives a
Gaussian product directly.  In both cases the same completion exposes the
higher-order remainder to which the truncation theorem applies.

\subsection*{Main results}

We first state the matching result.  Let $K_n$ be the complete graph on $n$
vertices.  If $A$ is a symmetric matrix of even order $n$, its hafnian is
\[
  \haf(A)=
  \sum_{M}\prod_{uv\in M}A_{uv},
\]
where the sum is over the perfect matchings $M$ of $K_n$.  The support graph
of a nonnegative symmetric matrix $A$ has an edge $uv$ exactly when
$A_{uv}>0$.  For a graph $G$, write $\delta(G)$ for its minimum degree and
$\#\PM(G)$ for its number of perfect matchings.  Hence $\#\PM(G)$ is the
hafnian of the adjacency matrix of $G$.

For the bipartite problem, let $S_m$ be the permutations of
$\{1,\ldots,m\}$.  The permanent of an $m$ by $m$ matrix $B=(B_{ij})$ is
\[
  \operatorname{per}(B)
  =\sum_{\sigma\in S_m}\prod_{i=1}^m B_{i,\sigma(i)}.
\]
The bipartite support of a nonnegative matrix has an edge $ij$ exactly when
$B_{ij}>0$.  For a zero-one matrix, the permanent counts the perfect
matchings of this support graph.

We measure multiplicative error by $e^{\pm\eps}$.  Thus an estimate
$\widehat Q$ of a positive quantity $Q$ is acceptable when
$e^{-\eps}Q\le\widehat Q\le e^\eps Q$.  This convention is equivalent to
$1\pm O(\eps)$ for $0<\eps<1$.  All numerical inputs, including $\eps$, are
encoded in binary.  For a rational matrix $C$, let $L_C$ denote its encoding
length, and let $L_\eps$ denote the encoding length of $\eps$.

\begin{theorem}[Main result for matchings]
\label{thm:main}
Fix $0<\gamma<1/2$ and $0<\theta\le1$.  The following deterministic
algorithms exist.
\begin{enumerate}
\item[\rm(a)]
Let $A$ be a rational symmetric $n$ by $n$ matrix, where $n$ is even, with
zero diagonal and all off-diagonal entries in
$\{0\}\cup[\theta,1]$.  Suppose that its support graph $G$ satisfies
\begin{equation}
  \delta(G)\ge\left(\frac12+\gamma\right)n,
  \label{eq:main-degree}
\end{equation}
and let $0<\eps<1$.  The algorithm returns an estimate $\widehat H>0$ such
that
\[
  e^{-\eps}\haf(A)\le \widehat H\le e^\eps\haf(A).
\]
Its bit complexity is
$n^{O_{\gamma,\theta}(1)}\eps^{-O_{\gamma,\theta}(1)}
\poly(L_A+L_\eps)$.

\item[\rm(b)]
Let $B$ be a rational $m$ by $m$ matrix with entries in
$\{0\}\cup[\theta,1]$.  Suppose that every row and column of its support
contains at least $(1/2+\gamma)m$ entries, and let $0<\eps<1$.  The
algorithm returns an estimate $\widehat Z_B>0$ such that
\[
  e^{-\eps}\operatorname{per}(B)
  \le\widehat Z_B
  \le e^\eps\operatorname{per}(B).
\]
Its bit complexity is
$m^{O_{\gamma,\theta}(1)}\eps^{-O_{\gamma,\theta}(1)}
\poly(L_B+L_\eps)$.
\end{enumerate}
\end{theorem}

Taking $\theta=1$ gives a deterministic FPTAS for unweighted perfect
matchings in the corresponding graph classes.  Dirac's theorem guarantees a
perfect matching at minimum degree $n/2$ \cite{Dirac1952}.  The fixed margin
in \cref{thm:main} is therefore measured from the natural existence
threshold, although it still allows each vertex to miss nearly half of the
other vertices.

We next state the Ising result.  Let $J$ be a real symmetric $n$ by $n$
matrix with zero diagonal.  For a spin configuration
$\sigma\in\{\pm1\}^n$, its interaction energy is
$\frac12\sigma^{\mathsf T}J\sigma$, and the zero-field partition function is
$Z(J)=\sum_{\sigma\in\{\pm1\}^n}
\exp(\sigma^{\mathsf T}J\sigma/2)$.
We write $\lambda_{\max}(J)$ for the largest eigenvalue of $J$, use $I$ for
the identity matrix, and let $\|\cdot\|$ denote the spectral operator norm.

\begin{theorem}[Diffuse zero-field Ising model]
\label{thm:ising}
Fix $\beta>0$ and $0<\kappa\le1$.  Let $J$ be a real symmetric $n$ by $n$
matrix with rational entries and zero diagonal, and suppose that
\[
  \max_{i,j}|J_{ij}|\le\frac\beta n,
  \qquad
  \lambda_{\max}(J)\le1-\kappa.
\]
Given $0<\eps<1$, there is a deterministic algorithm that returns a positive
estimate $\widehat Z$ satisfying
\[
  e^{-\eps}Z(J)\le\widehat Z\le e^\eps Z(J).
\]
Its bit complexity is
$n^{O_{\beta,\kappa}(1)}\eps^{-O_{\beta,\kappa}(1)}
\poly(L_J+L_\eps)$.
Moreover, uniformly over all real matrices satisfying the two displayed
hypotheses,
\[
  Z(J)=2^n\det(I-J)^{-1/2}
       \bigl(1+O_{\beta,\kappa}(1/n)\bigr),
\]
where the square root is positive.  In particular,
\(\log Z(J)=n\log2-\frac12\log\det(I-J)
+O_{\beta,\kappa}(1/n)\).  The determinant term satisfies
$|\frac12\log\det(I-J)|\le\beta^2/(4\kappa)$.
\end{theorem}

Rationality in \cref{thm:ising} is required only for the algorithmic
statement.

The spectral assumption in \cref{thm:ising} is one-sided.  It keeps the top
eigenvalue a fixed distance below the Gaussian singularity at $1$, but it
does not require $\|J\|<1$.  Negative eigenvalues may therefore have magnitude
larger than one.  The entrywise hypothesis is separate.  It is the
diffuseness condition that makes the coordinate-subset expansion summable.

We record three consequences of the matching theorem.  The first removes
support zeros altogether.

\begin{corollary}[Full-support inputs with weights bounded away from zero]
\label{cor:positive-dense}
Fix $0<\theta\le1$.  There is a deterministic FPTAS for the hafnian of every
even-dimensional rational symmetric matrix whose diagonal is zero and whose
off-diagonal entries lie in $[\theta,1]$, and for the permanent of every
rational square matrix whose entries lie in $[\theta,1]$.
\end{corollary}

The same algorithm treats odd-order graphs by leaving one vertex unmatched.

\begin{corollary}[Near-perfect matchings]
\label{cor:near-perfect}
Fix $0<\gamma<1/2$.  There is a deterministic FPTAS for the number of
near-perfect matchings in every odd-order graph $G$ satisfying
$\delta(G)\ge(1/2+\gamma)|V(G)|$.
\end{corollary}

The interval $[\theta,1]$ is only a normalization.  If all positive entries
lie in $[c_0,C_0]$ for fixed positive endpoints, divide the matrix by $C_0$
and restore the resulting homogeneous factor at the end.

The third consequence identifies the leading exponential contribution to a
weighted hafnian.  Suppose that $A=(a_{uv})$ satisfies part~\textup{(a)} of
\cref{thm:main}, and let $G$ be its support graph.  Define the fractional
perfect-matching polytope
\[
  \mathcal D(G)=
  \left\{w\in\mathbb R_{\ge0}^{E(G)}:
    \sum_{e\ni v}w_e=1\text{ for every }v\right\}.
\]
Thus a point of $\mathcal D(G)$ assigns nonnegative weights to the edges,
with total weight one at each vertex.  Its weighted entropy is
\[
  h_A(G)=
  \max_{w\in\mathcal D(G)}
  \sum_{uv\in E(G)}w_{uv}\log\frac{a_{uv}}{w_{uv}}.
\]
All logarithms are natural, and a zero summand is interpreted as zero.  For
the adjacency matrix of $G$, this is the entropy $h(G)$ used by Cuckler and
Kahn.

\begin{corollary}[Entropy formula]
\label{cor:constant-scale-entropy}
Fix $0<\gamma<1/2$ and $0<\theta\le1$.  There are constants
$0<a_{\gamma,\theta}\le b_{\gamma,\theta}<\infty$ such that every matrix
$A$ satisfying part~\textup{(a)} of \cref{thm:main} satisfies
\[
  a_{\gamma,\theta}
  \le
  \frac{\haf(A)}
  {e^{h_A(G)}(n-1)!!(n-1)^{-n/2}}
  \le
  b_{\gamma,\theta}.
\]
Here $(n-1)!!=1\cdot3\cdots(n-1)$.  Equivalently,
$\log\haf(A)=h_A(G)-n/2+O_{\gamma,\theta}(1)$,
where the remainder is bounded independently of $n$ and $A$.
\end{corollary}

For unweighted Dirac graphs, Cuckler and Kahn determine the logarithm of the
number of perfect matchings up to an $o(n)$ error
\cite{CucklerKahn2009,CucklerKahnHamiltonian2009}.  In the fixed-margin
regime, \cref{cor:constant-scale-entropy} improves this to
$O_\gamma(1)$ and extends the formula to weights in a fixed positive
interval.  The algorithms allow their constants and running-time exponents
to depend on the fixed parameters.  We do not optimize this dependence.

\subsection*{Technical overview}

The following overview traces both counting problems to one common
truncation theorem, then explains the two estimates in its proof.  The
argument is organized into six stages, with the fourth divided into its two
estimates.

\overviewstage{1}{Entropy and the inverse-Gamma reduction}
Let $A$ be an input to part~\textup{(a)} of \cref{thm:main}.  Maximizing the
weighted entropy gives a unique fractional perfect matching $w^*$.  The
Lagrange-multiplier equations factor it as
$w^*_{uv}=a_{uv}r_ur_v$ for positive vertex factors $r_v$.  Every perfect
matching uses each factor once, so
\[
  \haf(A)=e^{h_A(G)}\haf(X),
  \qquad X_{uv}=a_{uv}r_ur_v.
\]
The matrix $X$ is symmetric and stochastic.  Let $\mathbf J_n$ and $I_n$
be the all-ones and identity matrices, and let $\1$ be the all-ones column
vector.  Put
$P_0=(\mathbf J_n-I_n)/(n-1)$, and write $E=X-P_0$.  The degree and weight
margins make $E$ diffuse and supply a fixed spectral margin.  Since $X$ and
$P_0$ are stochastic, $E\1=0$.  Thus only the normalized ratio
$Q_E=\haf(P_0+E)/\haf(P_0)$ remains to be computed.

Expand the numerator according to the edges chosen from $E$.  The coefficient
left by the reference matrix depends only on the number of chosen edges and
is a negative moment of one Gamma variable.  More precisely, if $U$ has
shape $(n+1)/2$ and rate $(n-1)/2$, then
\[
  Q_E=\mathbb E\,Z_E(U^{-1}),
  \qquad
  Z_E(t)=\sum_M t^{|M|}\prod_{uv\in M}E_{uv},
\]
where $M$ ranges over all matchings.  The random variable $U$ is concentrated
near its mode $1$, and the contribution outside a fixed interval around $1$
is exponentially small.  The bipartite argument begins with Sinkhorn scaling
and applies the same construction after symmetric dilation, using the
complete-bipartite reference and slightly different Gamma parameters.

\overviewstage{2}{The Hubbard--Stratonovich reduction}
The starting point for \cref{thm:ising} is the Hubbard--Stratonovich identity,
which removes the spin sum.  If $\psi$ is the canonical complex Gaussian
vector with bilinear moment matrix $J$, and $\E_J$ denotes expectation with
respect to this vector, summing the resulting exponential moments over the
spins gives
\[
  Z(J)=2^n\E_J\prod_{i=1}^n\cosh(\psi_i).
\]
Unlike the matching problem, this reduction needs neither an entropy scaling
nor an auxiliary Gamma average.  It does not yet have the higher-order
remainder required by the truncation theorem.

\overviewstage{3}{Quadratic resummation}
The matching polynomial also has a Gaussian product representation.  If
$\xi$ has bilinear second-moment matrix $E$, then the Wick--Isserlis formula
\cite{Isserlis1918} gives
$Z_E(t)=\mathbb E\prod_i(1+\sqrt t\,\xi_i)$.  A direct expansion cannot be
truncated.  A term supported on $s$ coordinates has the rough size
$(Cs/n)^{s/2}$, while there are about $(en/s)^s$ possible supports.  The
centering relation $E\1=0$ creates cancellation between the raw
layers, but taking their absolute values destroys it.  To repair this, we
remove the first two Taylor terms exactly.  The centering relation
kills the linear term, and the quadratic term is absorbed into the Gaussian
density.  The same completion removes the quadratic term of $\cosh z$ in the
Ising product.  For a real symmetric moment matrix $L$, let $\phi$ denote
the canonical complex Gaussian vector with bilinear moments $L$.  We write
$\mathcal Z_L=\E_L\prod_i g(\phi_i)$, using the relevant one-coordinate
factor in each application.  The two exact identities are
\[
  \begin{aligned}
  Z_E(t)&=\det(I+tE)^{-1/2}\mathcal Z_{K_t},
    &K_t&=tE(I+tE)^{-1},\\
  Z(J)&=2^n\det(I-J)^{-1/2}\mathcal Z_K,
    &K&=J(I-J)^{-1}.
  \end{aligned}
\]
The corresponding one-coordinate factors are
\[
  g_{\mathrm{pm}}(z)=(1+z)e^{-z+z^2/2},
  \qquad
  g_{\mathrm{Is}}(z)=\cosh(z)e^{-z^2/2}.
\]
Now $g_{\mathrm{pm}}(z)-1=O(z^3)$.  In the Ising identity the same quadratic
resummation leaves
$g_{\mathrm{Is}}(z)-1=-z^4/12+O(z^6)$.  Completion also preserves the
entrywise scale and the required spectral margin.  Consequently, both
applications enter the common theorem, with vanishing order three for
matchings and four for Ising.

For either factor, put $f=g-1$ and write $[n]=\{1,\ldots,n\}$.  If
$S\subseteq[n]$, let
\[
  W_K(S)=\E_{K[S]}\prod_{i\in S}f(\phi_i),
  \qquad W_K(\varnothing)=1,
\]
where $K[S]$ is the principal submatrix indexed by $S$.

\begin{theorem*}[\Cref{thm:gaussian-product-truncation}, informal]
Suppose that $g$ has a quadratic Gaussian envelope, that
$f(z)=O(z^\ell)$ for some $\ell\ge3$, and that, for fixed $\beta,d>0$, the
real symmetric moment matrix $K$ satisfies
$\max_{i,j}|K_{ij}|\le\beta/n$ and the spectral margin condition of
\cref{sec:gaussian-principle} with constant $d$.
Put $\alpha=\ell/2-1$.  Then there are constants $c,C,a>0$ such that
\[
  \mathcal Z_K=\sum_{S\subseteq[n]}W_K(S),
  \qquad
  \left|\mathcal Z_K-\sum_{|S|<R}W_K(S)\right|
  \le\left(C\left(\frac Rn\right)^\alpha\right)^R+e^{-an}
\]
for every $2\le R\le cn$ and all sufficiently large $n$.
Moreover, $\mathcal Z_K=1+O(1/n)$ uniformly over the same class.
\end{theorem*}

The formal statement makes the envelope and spectral margin precise.  It
does not require $K$ to be positive semidefinite, the coefficients to have
one sign, or an associated polynomial to be zero-free.  Once the input
hypotheses produce an admissible pair $(g,K)$, the tail analysis is common
to both applications.

\overviewstage{4(a)}{Small coordinate sets}
Let $S$ have size $s$ below a small fixed multiple of $n$.  To bound
$W_K(S)$, consider the complex dilation
$\Phi_S(\zeta)=\mathbb E_{K[S]}\prod_{i\in S}f(\zeta\phi_i)$.
The entrywise bound on $K$ makes the row norm of $K[S]$ at most
$\beta s/n$.  Hence $\Phi_S$ is holomorphic and uniformly bounded on a disk
of radius comparable with $\sqrt{n/s}$.  Because $f$ vanishes to order
$\ell$, the function $\Phi_S$ vanishes to order $\ell s$ at the origin.
The maximum-modulus principle then gives
\[
  |W_K(S)|\le\left(\frac{Cs}{n}\right)^{\ell s/2}.
\]
Summing over all $S$ of size $s$ gives the layer bound
$(C(s/n)^\alpha)^s$, where $\alpha=\ell/2-1>0$.  This proves the first term
in the truncation estimate as long as the support remains below a small
linear cutoff.

\overviewstage{4(b)}{Recombining the large tail}
At a linear cutoff the dilation radius is only a constant, so the preceding
termwise bound no longer pays for all supports.  This is the main difficulty.
We return to the single Gaussian integral and use
\[
  \sum_{|S|\ge R}\prod_{i\in S}f(\phi_i)
  =\prod_{i=1}^n g(\phi_i)
   -\sum_{|S|<R}\prod_{i\in S}f(\phi_i).
\]
The absolute value is taken only after this recombination.  We then separate
the event on which few coordinates of $\phi$ have magnitude greater than a
small fixed constant $\delta>0$ from the event on which many do.

On the first event, every large support contains many small coordinates, and
the local estimate
$|f(z)|\le C_f\delta^{\ell-2}|z|^2$ gives an exponential contraction.  An
elementary-symmetric-function bound makes this contraction explicit, and
$\delta$ is chosen so that it dominates the remaining combinatorial factor.
The quadratic envelope is integrable by the spectral margin.

On the second event, some set $T$ of $r$ coordinates carries a fixed amount
of Gaussian energy.  The quadratic form measuring this energy is controlled
by the principal submatrix $|K|[T]$, where $|K|$ is the spectral absolute
value of $K$.  Diffuseness yields the estimate
\[
  \bigl\||K|[T]\bigr\|\le\beta\sqrt{r/n},
\]
which permits an exponential tilt of order $\sqrt{n/r}$.  For some fixed
$a>0$, a fixed witness set then costs $\exp\{-a\sqrt{rn}\}$.  Write
$\rho=r/n$.  When $\rho$ is a small positive constant, this gain has rate
$\sqrt\rho$, whereas the number of witness sets has entropy rate
$\rho\log(1/\rho)$.  Since
$\rho\log(1/\rho)=o(\sqrt\rho)$, choosing the linear cutoff small enough
makes the total contribution exponentially small.  This is why the present
proof uses more than a global spectral bound, and why the entrywise hypothesis
is essential to it.

\overviewstage{5}{Evaluating retained subsets}
It remains to compute $W_K(S)$ for $|S|<R$.  Since $f=g-1$,
inclusion--exclusion expresses it through full-product expectations on
subsets $T\subseteq S$.  For matchings, each such expectation is a weighted
monomer--dimer partition function and is evaluated by vertex deletion.  For
Ising, it is an explicit sum over the $2^{|T|}$ spin configurations on $T$.
In both cases, one coefficient $W_K(S)$ costs
$3^{|S|}\poly(|S|)$ arithmetic operations.

\overviewstage{6}{From truncation to an FPTAS}
Let $b$ be of order $\log_2(1/\eps)$, and fix a sufficiently large constant
$C_{\mathrm{cut}}$.  We choose the least integer $R\ge2$ for which
\[
  \left(C_{\mathrm{cut}}
        \left(\frac Rn\right)^\alpha\right)^R\le2^{-b}.
\]
Minimality gives
$R\log(en/R)=O(b+\log n)$, and therefore the sum of the retained costs is
$n^{O(1)}2^{O(b)}$.  If one simply took $R$ proportional to $b$, coefficient
enumeration would instead have the quasipolynomial form
$n^{O(\log(1/\eps))}$.  The least-cutoff balance is what turns the truncation
estimate into a fully polynomial scheme.  When $b$ is a fixed positive
multiple of $n$, direct enumeration costs $2^{O(n)}=2^{O(b)}$ and supplies
the complementary branch.

The zero-free interpolation algorithms cited below use logarithmically many
Taylor coefficients, but their dense-instance running times are
quasipolynomial \cite{Barvinok2017,BarvinokBarvinok2021}.  Our completion
resums the determinant contribution exactly, and the remaining
$(R/n)^{\alpha R}$ decay makes the adaptive enumeration polynomial.

Every retained term is an explicit formula involving matrices that remain a
fixed distance from singularity, so polynomially many working bits suffice.
The details, including the entropy scaling, are given in
\cref{app:bit-complexity}.  The Ising application needs no scaling step.

\subsection*{Previous work}

\emph{Randomized matching algorithms.}
Exact counting is $\#\mathrm P$-complete already for zero-one permanents
\cite{Valiant1979}, and remains $\#\mathrm P$-complete on very dense
bipartite and general graphs \cite{ElMaaloulyWang2022}.  Jerrum and Sinclair
gave an FPRAS for graphs of even order $n$ and minimum degree at least $n/2$
\cite{JerrumSinclair1989}, and
Jerrum, Sinclair, and Vigoda gave an FPRAS for the permanent of every
nonnegative matrix \cite{JerrumSinclairVigoda2004}.  Chen, Vigoda, and Yang
recently improved the running time in this line of work
\cite{ChenVigodaYang2026}.  For general nonbipartite graphs, an FPRAS for the
number of perfect matchings remains open.  Ebrahimnejad, Nagda, and Oveis
Gharan obtained randomized approximate counting and sampling for regular
strong expanders \cite{EbrahimnejadNagdaOveisGharan2022}.

\emph{Deterministic matching algorithms.}
For matrices with entries in $[\theta,1]$, Barvinok gives a deterministic
quasipolynomial-time relative approximation for permanents and hafnians
\cite{Barvinok2017}.  The same work treats zero-one permanents with a small
fixed fraction of zeros in each row and column.  The full-support case of
\cref{thm:main} makes the running time fully polynomial, and its support
version reaches every fixed positive margin above the Dirac threshold.  For
arbitrary nonnegative permanents, deterministic methods give much coarser
universal factors.  One example is the tight $(\sqrt2)^m$ guarantee for the
Bethe permanent \cite{AnariRezaei2025}.  These results apply much more
broadly, but they do not give arbitrary relative accuracy.

\emph{Dense matching asymptotics.}
Cuckler and Kahn proved
$\log\#\PM(G)=h(G)-n/2+o(n)$ for every even-order Dirac graph
\cite{CucklerKahn2009,CucklerKahnHamiltonian2009}.
\Cref{cor:constant-scale-entropy} replaces $o(n)$ by a bounded remainder on
the fixed-margin class.  McCullagh obtained a determinantal approximation
with relative error $O(1/m)$ for moderate-deviation sequences of doubly
stochastic matrices \cite{McCullagh2014}.  Koehler and Leung prove
high-probability zero-free regions for permanents of an all-ones matrix under
independent centered perturbations \cite{KoehlerLeung2026}.  Their theorem is
an average-case bipartite result, while the hypotheses here are deterministic
and also cover nonbipartite hafnians.

\emph{Ising algorithms.}
Jerrum and Sinclair gave an FPRAS for ferromagnetic Ising systems
\cite{JerrumSinclair1993}.  Eldan, Koehler, and Zeitouni proved rapid Glauber
mixing for general interactions under the spectral condition $\|J\|<1$
\cite{EldanKoehlerZeitouni2022}.  On the diffuse part of that norm regime,
\cref{thm:ising} gives deterministic relative approximation, and its
one-sided condition also permits negative eigenvalues of magnitude larger
than one.
On the deterministic side, complex-zero methods give an FPTAS for
bounded-degree zero-field Ising models throughout the correlation-decay region
\cite{LiuSinclairSrivastava2019}.  More generally, Patel and Regts turn
zero-free regions into deterministic polynomial-time approximation algorithms
for a broad class of bounded-degree graph polynomials
\cite{PatelRegts2017}.  Their coefficient computation enumerates bounded-size
connected subgraphs and relies on the maximum degree being bounded.  Barvinok
and Barvinok give a
quasipolynomial relative approximation under a global bound on the total
absolute interaction incident to each spin \cite{BarvinokBarvinok2021}.
On dense instances, its Taylor truncation has geometric decay and
quasipolynomial coefficient-enumeration cost.
Mean-field methods apply much more broadly to dense models, but give additive
control of the free energy rather than relative control of the partition
function \cite{JainKoehlerMossel2018}.

The paper is organized as follows.  \Cref{sec:gaussian-principle} states and
proves the common truncation theorem.
\Cref{sec:matching} develops the matching and permanent applications.
The Ising application is
proved in \cref{sec:ising}.  We return to the limitations of the method and
the problems they leave open in \cref{sec:scope}.  The appendices contain the
deferred scaling, analytic, and bit-complexity estimates.

\section{A truncation principle for diffuse Gaussian products}
\label{sec:gaussian-principle}

\subsection{Setup and theorem}

Although the matching and Ising reductions begin differently, after
quadratic completion they produce the same analytic object: the expectation
of a product $\prod_i g(\phi_i)$, where $g$ is entire and the Gaussian
bilinear moment matrix has entries of order $1/n$.  The moment matrix need
not be positive semidefinite, so the coordinates of $\phi$ may be complex.
Wick's rule still expresses their moments in terms of this matrix, but
absolute convergence must be controlled separately.  The quadratic envelope
and spectral margin introduced below provide that control uniformly over the
admissible class.  The theorem of
this section shows that if $g-1$ vanishes to order at least three, then the
exact expansion by coordinate support has a rapidly decaying tail.  In the
two applications, the matching factor grows in the real direction and the
Ising factor in the imaginary one.

For a real
symmetric matrix $K$, write $K=K_+-K_-$ for its positive and negative
parts, and put
\[
  |K|=K_++K_-,
  \qquad
  A_K=K_+^{1/2}+iK_-^{1/2}.
\]
The matrix $A_K$ is complex symmetric, $A_K^2=K$, and
$A_K^*A_K=|K|$.  If $x$ is a standard real Gaussian vector, we write
$\phi=A_Kx$.  Thus $\E\phi\phi^{\mathsf T}=K$ and
\[
  \sum_i(\operatorname{Re}\phi_i)^2=x^{\mathsf T}K_+x,
  \qquad
  \sum_i(\operatorname{Im}\phi_i)^2=x^{\mathsf T}K_-x.
\]
Whenever the integral converges absolutely, let
$\E_KF(\phi)=\E_xF(A_Kx)$.  For a positive integer $n$, write
$[n]=\{1,\ldots,n\}$.  For $S\subseteq[n]$, let $K[S]$ denote the
principal submatrix indexed by $S$.  For a matrix $C$, write
$\|C\|_{\mathrm{row}}=\max_i\sum_j|C_{ij}|$.  We use
\[
  h_2(u)=-u\log u-(1-u)\log(1-u)
\]
for the binary entropy, with the usual continuous interpretation at zero
and one.

Together with the spectral margin imposed below, the global bound in the next
definition makes the product integrable, with enough room left for the
exponential tilts used below.  The local
vanishing order counts how many Gaussian legs each selected coordinate must
carry.

\begin{definition}[Admissible scalar factor]
\label{def:admissible-factor}
Let $a_{\mathrm R},a_{\mathrm I}\ge0$ be two envelope weights, not both
zero.  Let $\ell\ge3$ be an integer, let $C_f>0$, and let
$0<\delta_f\le1$.  We call an entire function $g$ with $g(0)=1$
\emph{admissible with these parameters} if
\begin{align}
  |g(z)|
  &\le
  \exp\!\left\{a_{\mathrm R}(\operatorname{Re}z)^2
                 +a_{\mathrm I}(\operatorname{Im}z)^2\right\}
  &&(z\in\C),
  \label{eq:factor-global-envelope}\\
  |g(z)-1|&\le C_f|z|^\ell
  &&(|z|\le\delta_f).
  \label{eq:factor-local-vanishing}
\end{align}
Put $f=g-1$ and
$a_{\max}=\max\{a_{\mathrm R},a_{\mathrm I}\}$.
\end{definition}

The two factors used later are
$g_{\mathrm{pm}}(z)=(1+z)e^{-z+z^2/2}$, with envelope weights
$(a_{\mathrm R},a_{\mathrm I})=(1,0)$ and vanishing order three, and
$g_{\mathrm{Is}}(z)=\cosh(z)e^{-z^2/2}$, with weights $(0,1/2)$ and
vanishing order four.  Their precise bounds are stated in the corresponding
applications and proved in Appendix~\ref{app:analytic-estimates}.

For a fixed scalar factor, the envelope matrix below is the quadratic form
that bounds the absolute value of the full product.  Its spectral margin is
exactly what keeps that bound integrable.  Notice that the condition is
one-sided whenever one of the two envelope weights is zero.

\begin{definition}[Admissible Gaussian moment matrix]
\label{def:admissible-kernel}
Fix the envelope weights $a_{\mathrm R}$ and $a_{\mathrm I}$.  For a real
symmetric matrix $K$, define the positive-semidefinite envelope matrix
\[
  K_{\mathrm{env}}=a_{\mathrm R}K_++a_{\mathrm I}K_-.
\]
Let $\beta\ge0$ and $0<d\le1$.  We call a real symmetric $n$ by $n$
matrix $K$ \emph{$(\beta,d)$-admissible} if
\begin{equation}
  \max_{i,j}|K_{ij}|\le\frac\beta n,
  \qquad
  I-2K_{\mathrm{env}}\succeq dI.
  \label{eq:diffuse-covariance-assumptions}
\end{equation}
\end{definition}

Given an admissible scalar factor $g$ and a $(\beta,d)$-admissible
matrix $K$ of dimension $n$, retain the notation $f=g-1$.  For
$S\subseteq[n]$, define
\[
  W_K(S)=\E_{K[S]}\prod_{i\in S}f(\phi_i),
  \qquad
  W_K(\varnothing)=1,
  \qquad
  \mathcal Z_K=\E_K\prod_{i=1}^ng(\phi_i).
\]

\begin{theorem}[Diffuse Gaussian-product truncation]
\label{thm:gaussian-product-truncation}
Fix $a_{\mathrm R}$, $a_{\mathrm I}$, $\ell$, $C_f$, $\delta_f$, $\beta$,
and $d$ as in
\cref{def:admissible-factor,def:admissible-kernel},
and set
\[
  \alpha=\frac{\ell}{2}-1\ge\frac12.
\]
There are positive constants
$c_{\mathrm{tr}}$, $a_{\mathrm{tail}}$, $C_0$, and $C_1$, and an integer
$n_0$, depending only on these fixed parameters, with the following
properties.  Suppose that $g$ is admissible with the stated scalar
parameters and that $K$ is a $(\beta,d)$-admissible matrix of dimension
$n$.  Then all the expectations below converge absolutely and
\begin{equation}
  \mathcal Z_K=\sum_{S\subseteq[n]}W_K(S).
  \label{eq:gaussian-subset-expansion}
\end{equation}
If $n\ge n_0$ and $R$ is an integer with
$2\le R\le c_{\mathrm{tr}}n$, then
\begin{equation}
  \left|\mathcal Z_K-\sum_{|S|<R}W_K(S)\right|
  \le
  \left(C_1\left(\frac Rn\right)^\alpha\right)^R
  +e^{-a_{\mathrm{tail}}n}.
  \label{eq:diffuse-gaussian-truncation-bound}
\end{equation}
Moreover,
\begin{equation}
  |\mathcal Z_K-1|\le\frac{C_0}{n}
  \qquad(n\ge n_0).
  \label{eq:diffuse-gaussian-near-one}
\end{equation}
When $\beta>0$, the constant $c_{\mathrm{tr}}$ may be chosen so that
$c_{\mathrm{tr}}\le1/(8\beta)$.  Consequently,
$\|K[S]\|_{\mathrm{row}}\le1/8$ whenever
$|S|\le c_{\mathrm{tr}}n$.
\end{theorem}

\subsection{Scalar estimates, marginalization, and small coordinate sets}

We prove the theorem in several steps.  The following elementary remark records
three consequences of the hypotheses on $g$.

\begin{remark}[Useful scalar bounds]
\label{rem:factor-consequences}
Admissibility immediately gives, for $0<\delta\le\delta_f$,
\begin{align*}
  1+|f(z)|&\le3e^{a_{\mathrm R}(\operatorname{Re}z)^2+
    a_{\mathrm I}(\operatorname{Im}z)^2},
  \qquad |f(z)|\le C_f\delta^{\ell-2}|z|^2\quad(|z|\le\delta),\\
  a_{\mathrm R}(\operatorname{Re}(\zeta w))^2
  +a_{\mathrm I}(\operatorname{Im}(\zeta w))^2
  &\le a_{\max}|\zeta|^2|w|^2.
\end{align*}
We use these three consequences without further comment.
\end{remark}

The next lemma justifies restricting an indefinite complex Gaussian to a
set of coordinates.  Matching bilinear moments alone do not immediately
give this conclusion for an arbitrary entire function, so we connect the
two representations by analytic continuation.

\begin{lemma}[Complex Gaussian marginalization]
\label{lem:gaussian-marginalization}
Fix the envelope weights $a_{\mathrm R}$ and $a_{\mathrm I}$.  Let $K$ be
$(\beta,d)$-admissible, let $S\subseteq[n]$, and let
$F:\C^S\to\C$ be entire.  Suppose that
\[
  |F(z)|\le
  A\exp\!\left\{
    \sum_{i\in S}
    \bigl(a_{\mathrm R}(\operatorname{Re}z_i)^2
          +a_{\mathrm I}(\operatorname{Im}z_i)^2\bigr)
  \right\}
\]
for some $A<\infty$.  Then both expectations converge absolutely and
\[
  \E_KF((\phi_i)_{i\in S})=\E_{K[S]}F(\phi).
\]
\end{lemma}

The proof dilates both Gaussian representations by a complex parameter
$\zeta$.  Interlacing transfers the spectral margin to $K[S]$, giving a
common holomorphic neighborhood of $[0,1]$, while Wick's rule identifies the
Taylor series of the two representations at zero.  The details are given in
\cref{subsec:gaussian-marginalization-proof}.

We next control a contribution whose support is smaller than a fixed
fraction of all coordinates.  The guiding count is simple.  If a support
has size $s$, an order-$\ell$ zero forces its Gaussian expansion to begin in
degree $\ell s$.  Pairing those legs costs roughly $1/n$ per edge, while the
number of possible pairings contributes the corresponding power of $s$, up
to constants depending only on $\ell$.  This suggests a factor
$(s/n)^{\ell s/2}$.  Complex dilation proves this estimate without enumerating
the pairings.

\begin{lemma}[Small supports]
\label{lem:small-subset-general}
Let $g$ be an admissible scalar factor.  Let $K=(K_{ij})$ be a real
symmetric matrix of order $n$, and suppose that
$\max_{i,j}|K_{ij}|\le\beta/n$.
If $\beta=0$, then $K=0$ and $W_K(S)=0$ for every nonempty $S$.  Suppose
that $\beta>0$, and put
\[
  c_{\mathrm{loc}}=\frac1{64a_{\max}\beta},
  \qquad
  B=3(1-1/8)^{-1/2},
  \qquad
  C_{\mathrm{loc}}=16a_{\max}\beta B^{2/\ell}.
\]
If $S$ is nonempty and $s=|S|\le c_{\mathrm{loc}}n$, then the integral
defining $W_K(S)$ converges absolutely and
\begin{equation}
  |W_K(S)|
  \le\left(C_{\mathrm{loc}}\frac sn\right)^{\ell s/2}.
  \label{eq:small-support-bound}
\end{equation}
For a singleton, the stronger estimate
\begin{equation}
  |W_K(\{i\})|\le B(16a_{\max}\beta)^2n^{-2}
  \label{eq:singleton-support-bound}
\end{equation}
holds.
\end{lemma}

\begin{proof}
The case $\beta=0$ is immediate from the entrywise hypothesis.  Assume
$\beta>0$.  Put $L=K[S]$.  Then
$\|L\|\le\|L\|_{\mathrm{row}}\le\beta s/n$.
For a standard real Gaussian vector $x\in\R^s$, define
\[
  \Phi_S(\zeta)=
  \E_x\prod_{i\in S}f(\zeta(A_Lx)_i),
  \qquad
  \varrho^2=\frac{n}{16a_{\max}\beta s}.
\]
The assumption $s\le c_{\mathrm{loc}}n$ gives $\varrho\ge2$.  By
  \cref{rem:factor-consequences}, for $|\zeta|\le3\varrho/2$,
\[
  \prod_{i\in S}|f(\zeta(A_Lx)_i)|
  \le3^s\exp\{a_{\max}|\zeta|^2x^{\mathsf T}|L|x\}.
\]
Moreover, $2a_{\max}(3\varrho/2)^2\|L\|\le9/32<1$.
Thus the right-hand side gives a locally uniform integrable Gaussian
majorant on $|\zeta|<3\varrho/2$.  In particular, $\Phi_S$ is holomorphic
there and its value at $\zeta=1$ converges absolutely.

On $|\zeta|=\varrho$, we have
$2a_{\max}\varrho^2\|L\|\le1/8$.  Gaussian integration gives
\[
  |\Phi_S(\zeta)|
  \le3^s\det(I-2a_{\max}\varrho^2|L|)^{-1/2}
  \le3^s(1-1/8)^{-s/2}=B^s.
\]
Condition \eqref{eq:factor-local-vanishing}, together with analyticity,
says that $f$ has a zero of order at least $\ell$ at the origin.  Hence
$\Phi_S(\zeta)=\zeta^{\ell s}H_S(\zeta)$ for a holomorphic function $H_S$.
The maximum-modulus principle gives
\[
  |W_K(S)|=|\Phi_S(1)|
  \le \varrho^{-\ell s}B^s
  =B^s\left(16a_{\max}\beta\frac sn\right)^{\ell s/2},
\]
which is \eqref{eq:small-support-bound}.

When $S=\{i\}$, the substitution $x\mapsto-x$ shows that $\Phi_S$ is
even.  Its order of vanishing is therefore an even integer at least $\ell$,
hence at least four.  The same argument gives
$|W_K(\{i\})|\le B\varrho^{-4}=B(16a_{\max}\beta)^2n^{-2}$.
\end{proof}

\subsection{Recombined large-support tail}

The large-support estimate needs more information about small principal
submatrices of the Hermitian covariance $|K|$ than a global operator norm
provides.

\begin{lemma}[Envelopes and principal submatrices]
\label{lem:partial-product-envelope}
Let $K$ be a real symmetric $n$ by $n$ matrix, and let $\beta\ge0$ satisfy
$\max_{i,j}|K_{ij}|\le\beta/n$.  If $S\subseteq[n]$ has size $s$, then
\begin{align}
  \||K|[S]\|&\le\beta\sqrt{\frac sn},
  \label{eq:hermitian-principal-norm}\\
  |K|_{ii}&\le\frac\beta{\sqrt n},
  \qquad
  \operatorname{tr}|K|[S]\le\frac{\beta s}{\sqrt n},
  \qquad
  \operatorname{tr}|K|\le\beta\sqrt n.
  \label{eq:hermitian-principal-traces}
\end{align}
Suppose in addition that $g$ is admissible and $K$ is
$(\beta,d)$-admissible.  For $T\subseteq[n]$, let $\Pi_T$ be the coordinate
projection, and put
\[
  B_T=
  a_{\mathrm R}K_+^{1/2}\Pi_TK_+^{1/2}
  +a_{\mathrm I}K_-^{1/2}\Pi_TK_-^{1/2}.
\]
Then
\begin{equation}
  \prod_{i\in T}|g(\phi_i)|\le e^{x^{\mathsf T}B_Tx},
  \qquad
  0\preceq B_T\preceq K_{\mathrm{env}},
  \qquad
  \operatorname{tr}B_T\le a_{\max}\beta\sqrt n.
  \label{eq:partial-product-envelope}
\end{equation}
In particular, $I-2B_T\succeq dI$ and
\begin{equation}
  \E_x\prod_{i\in T}|g(\phi_i)|
  \le\det(I-2B_T)^{-1/2}
  \le e^{C\sqrt n},
  \label{eq:partial-product-integral}
\end{equation}
where $C$ depends only on $(a_{\max},\beta,d)$.
\end{lemma}

\begin{proof}
Let $\Pi_S$ be the coordinate projection onto $S$.  Since $|K|^2=K^2$,
\[
  (\Pi_S|K|\Pi_S)^2
  \preceq \Pi_S|K|^2\Pi_S
  =\Pi_SK^2\Pi_S.
\]
Indeed, the difference is
$\Pi_S|K|(I-\Pi_S)|K|\Pi_S\succeq0$.  Every entry of $K^2$ has modulus at
most $\sum_{k=1}^n|K_{ik}K_{kj}|\le\beta^2/n$.
Consequently,
$\|K^2[S]\|\le\|K^2[S]\|_{\mathrm{row}}\le\beta^2s/n$.  Taking
operator norms in the preceding positive-semidefinite inequality proves
\eqref{eq:hermitian-principal-norm}.

For every $i$,
$|K|_{ii}^2\le(|K|^2)_{ii}=(K^2)_{ii}\le\beta^2/n$.
This gives the diagonal estimate.  Summing it over $i\in S$, and then
over all coordinates, gives the two trace estimates.

For the envelope assertion, sum \eqref{eq:factor-global-envelope} over
$i\in T$.  The real-coordinate contribution is
$\sum_{i\in T}(\operatorname{Re}\phi_i)^2
=x^{\mathsf T}K_+^{1/2}\Pi_TK_+^{1/2}x$, and the imaginary contribution is
analogous.  This proves the first
inequality in \eqref{eq:partial-product-envelope}.  Since
$0\preceq\Pi_T\preceq I$, we have
$0\preceq B_T\preceq K_{\mathrm{env}}$.  Moreover,
$\operatorname{tr}B_T\le\operatorname{tr}K_{\mathrm{env}}
\le a_{\max}\operatorname{tr}|K|\le a_{\max}\beta\sqrt n$
by \eqref{eq:hermitian-principal-traces}.

The admissibility margin gives $I-2B_T\succeq dI$.  Gaussian integration
yields the first inequality in \eqref{eq:partial-product-integral}.  Finally,
for $0\preceq D\preceq(1-d)I$, the scalar inequality
$-\log(1-u)\le C_du$ on $0\le u\le1-d$ gives
$-\frac12\log\det(I-D)\le C_d\operatorname{tr}D$.
Apply this with $D=2B_T$ and use the trace bound above.
\end{proof}

We now estimate the event on which many Gaussian coordinates are large.
The required rate is stronger than an ordinary linear large-deviation bound.
If $r=\rho n$, a gain $e^{-c\rho n}$ can lose to the entropy
$e^{n h_2(\rho)}$, because $h_2(\rho)$ is of order
$\rho\log(1/\rho)$ near zero.  A witness-set estimate instead gives
$e^{-c\sqrt{\rho}\,n}$, which wins since
$\rho\log(1/\rho)=o(\sqrt\rho)$.  This square-root rate is the reason for
working with small principal submatrices of $|K|$.
For $1\le r\le n$ and $\delta>0$, let $\mathcal A_{r,\delta}$ be the event
that at least $r$ coordinates satisfy $|\phi_i|>\delta$.  For a measurable
event $\mathcal B\subseteq\R^n$, write
\[
  \mathcal J_T(\mathcal B)
  =\E_x\left[\ind{\mathcal B}\prod_{i\in T}|g(\phi_i)|\right].
\]

\begin{lemma}[Many large coordinates]
\label{lem:many-large-coordinates}
Let $g$ be admissible in the sense of
\cref{def:admissible-factor}.  Suppose that $K$ is
$(\beta,d)$-admissible and $\beta>0$.  There is a constant $C$, depending only on
$(a_{\max},\beta,d)$, such that, uniformly over $T\subseteq[n]$,
\begin{equation}
  \mathcal J_T(\mathcal A_{r,\delta})
  \le
  \exp\!\left\{
    -\frac{d\delta^2}{4\beta}\sqrt{rn}
    +nh_2(r/n)+C\sqrt n
  \right\}.
  \label{eq:many-large-coordinate-bound}
\end{equation}
\end{lemma}

\begin{proof}
If $\mathcal A_{r,\delta}$ occurs, some set $S$ of $r$ coordinates satisfies
$\sum_{i\in S}|\phi_i|^2\ge r\delta^2$.
For such a set, define
\[
  M_S=K_+^{1/2}\Pi_SK_+^{1/2}
      +K_-^{1/2}\Pi_SK_-^{1/2}.
\]
The coordinate energy on $S$ is $x^{\mathsf T}M_Sx$.  To estimate the
norm of $M_S$, let $R_S:\R^n\to\R^r$ be coordinate restriction and put
\[
  C_S=
  \begin{pmatrix}
    R_SK_+^{1/2}\\
    R_SK_-^{1/2}
  \end{pmatrix}.
\]
Then $C_S^{\mathsf T}C_S=M_S$.  Since
$K_+^{1/2}K_-^{1/2}=0$, the matrix $C_SC_S^{\mathsf T}$ is block
diagonal with blocks $K_+[S]$ and $K_-[S]$.  The nonzero eigenvalues of
$C_S^{\mathsf T}C_S$ and $C_SC_S^{\mathsf T}$ agree.  Because
$K_+[S],K_-[S]\preceq|K|[S]$, \cref{lem:partial-product-envelope} gives
\begin{equation}
  \|M_S\|\le\beta\sqrt{\frac rn},
  \qquad
  \operatorname{tr}M_S
  =\operatorname{tr}|K|[S]
  \le\frac{\beta r}{\sqrt n}.
  \label{eq:witness-matrix-bounds}
\end{equation}

Fix one witness set $S$, and put
$\lambda=d/(4\beta\sqrt{r/n})$.
By Markov's inequality and \cref{lem:partial-product-envelope}, the part of
$\mathcal J_T$ on which $x^{\mathsf T}M_Sx\ge r\delta^2$ is at most
\begin{equation}
  e^{-\lambda r\delta^2}
  \det(I-2B_T-2\lambda M_S)^{-1/2}.
  \label{eq:witness-tilt}
\end{equation}
Indeed, \eqref{eq:witness-matrix-bounds} gives
$2\lambda\|M_S\|\le d/2$, and $I-2B_T\succeq dI$.  Hence
\[
  I-2B_T-2\lambda M_S\succeq\frac d2I.
\]
The determinant estimate from \cref{lem:partial-product-envelope} applies to
$2B_T+2\lambda M_S$.  Its trace is bounded by
\[
  2a_{\max}\beta\sqrt n+2\lambda\frac{\beta r}{\sqrt n}
  =2a_{\max}\beta\sqrt n+\frac d2\sqrt r
  \le C\sqrt n.
\]
Thus \eqref{eq:witness-tilt} is at most
$\exp\{-d\delta^2\sqrt{rn}/(4\beta)+C\sqrt n\}$.
There are at most
$\binom nr\le\exp\{nh_2(r/n)\}$ witness sets.  A union bound proves
\eqref{eq:many-large-coordinate-bound}.
\end{proof}

The small-subset bound cannot be summed all the way to size $n$.  The
next lemma instead recombines all subsets beyond one linear cutoff before
taking an absolute value.

\begin{lemma}[Recombined large-support tail]
\label{lem:linear-cut-tail}
Fix the parameters in
\cref{def:admissible-factor,def:admissible-kernel}.
Let $g$ be admissible with the fixed scalar parameters, and let $K$ be a
$(\beta,d)$-admissible matrix of dimension $n$.  For every
$c_{\max}>0$ there are constants $0<c_*\le c_{\max}$ and
$a_{\mathrm{tail}}>0$, and an integer $n_{\mathrm{tail}}$, depending only
on the fixed parameters and $c_{\max}$, such that
\begin{equation}
  \left|\sum_{|S|\ge\lceil c_*n\rceil}W_K(S)\right|
  \le e^{-a_{\mathrm{tail}}n}
  \qquad(n\ge n_{\mathrm{tail}}).
  \label{eq:linear-cut-signed-tail}
\end{equation}
\end{lemma}

\begin{proof}
If $\beta=0$, then $K=0$ and every nonempty contribution vanishes.  We
may therefore assume that $\beta>0$.

Put $t_0=d/(8(\beta+1))$.  Choose $0<\delta\le\delta_f$ so small that
$\chi:=C_f\delta^{\ell-2}/t_0$ satisfies $3\chi<1$, and put
$a_{\mathrm{ld}}=d\delta^2/(4\beta)$.
We shall choose $0<c_*\le\min\{c_{\max},1/2\}$ after obtaining the two
estimates below.  For a fixed positive $c_*$, define
\[
  s_*=\lceil c_*n\rceil,
  \qquad
  r=\left\lceil\frac{s_*}{2}\right\rceil,
  \qquad
  \mathcal T_{s_*}(\phi)
  =\sum_{|S|\ge s_*}\prod_{i\in S}f(\phi_i).
\]
We take $n$ large enough that $s_*\ge2$.

First suppose that fewer than $r$ coordinates are larger than $\delta$ in
modulus.  On this event, put $H=\{i:|\phi_i|>\delta\}$ and
$A_{H^c}=\sum_{i\notin H}|f(\phi_i)|$.
Then $|H|<r$, so every subset of size at least $s_*$ uses at least
$k=s_*-r$ coordinates outside $H$.  If nonnegative numbers $a_i$ are indexed
outside $H$, their elementary symmetric polynomials satisfy
\[
  e_j((a_i))\le\frac{(\sum_i a_i)^j}{j!},
  \qquad
  \sum_{j\ge k}e_j((a_i))
  \le e^{\sum_i a_i}\frac{(\sum_i a_i)^k}{k!}.
\]
The first inequality follows by expanding $(\sum_i a_i)^j$, and the
second uses $k!/(k+j)!\le1/j!$ for $j\ge0$.

Apply this estimate with $a_i=|f(\phi_i)|$.  Summing first over the part
of a selected set contained in $H$ gives
\begin{equation}
  |\mathcal T_{s_*}(\phi)|
  \le
  \prod_{i\in H}(1+|f(\phi_i)|)
  e^{A_{H^c}}\frac{A_{H^c}^k}{k!}.
  \label{eq:few-large-pointwise}
\end{equation}
By \cref{rem:factor-consequences},
$A_{H^c}\le C_f\delta^{\ell-2}x^{\mathsf T}|K|x$.  The same remark and
$|H|<r$ give
$\prod_{i\in H}(1+|f(\phi_i)|)\le
3^r e^{x^{\mathsf T}K_{\mathrm{env}}x}$.  Finally, if
$Q=x^{\mathsf T}|K|x$, then $Q^k/k!\le t_0^{-k}e^{t_0Q}$.
Substitution in \eqref{eq:few-large-pointwise} yields
\begin{equation}
  |\mathcal T_{s_*}(\phi)|
  \le
  3^r\chi^{s_*-r}
  \exp\!\left\{
    x^{\mathsf T}
    \bigl(K_{\mathrm{env}}+(C_f\delta^{\ell-2}+t_0)|K|\bigr)x
  \right\}.
  \label{eq:few-large-envelope}
\end{equation}

Since $C_f\delta^{\ell-2}<t_0/3$ and $\|K\|\le\beta$,
\[
  2(C_f\delta^{\ell-2}+t_0)\|K\|
  \le\frac83t_0\beta
  =\frac{d\beta}{3(\beta+1)}
  \le\frac d3.
\]
Together with $I-2K_{\mathrm{env}}\succeq dI$, this gives
\[
  I-2K_{\mathrm{env}}-2(C_f\delta^{\ell-2}+t_0)|K|
  \succeq\frac d2I.
\]
The trace of the positive-semidefinite matrix in the exponent of
\eqref{eq:few-large-envelope} is at most
$(a_{\max}+C_f\delta^{\ell-2}+t_0)\operatorname{tr}|K|\le C\sqrt n$.
Gaussian integration and the determinant estimate used in
\cref{lem:partial-product-envelope} therefore give
\begin{equation}
  \left|
  \E_x\left[
    \mathcal T_{s_*}(\phi)
    \ind{\mathcal A_{r,\delta}^{\mathsf c}}
  \right]
  \right|
  \le e^{C\sqrt n}3^r\chi^{s_*-r}.
  \label{eq:few-large-tail-bound}
\end{equation}

We next consider $\mathcal A_{r,\delta}$.  The exact pointwise identity
\begin{equation}
  \mathcal T_{s_*}(\phi)
  =\prod_{i=1}^ng(\phi_i)
   -\sum_{|S|<s_*}\prod_{i\in S}f(\phi_i)
  \label{eq:large-support-recombination}
\end{equation}
preserves the cancellation among all subsets beyond the cutoff.  For a
fixed set $S$,
\[
  \prod_{i\in S}|f(\phi_i)|
  \le\prod_{i\in S}(1+|g(\phi_i)|)
  =\sum_{T\subseteq S}\prod_{i\in T}|g(\phi_i)|.
\]
Apply \cref{lem:many-large-coordinates} to each partial product on the right, and also to
the full product in \eqref{eq:large-support-recombination}.  Uniformity in
$T$ gives
\begin{equation}
  \left|
  \E_x\left[
    \mathcal T_{s_*}(\phi)
    \ind{\mathcal A_{r,\delta}}
  \right]
  \right|
  \le
  e^{-a_{\mathrm{ld}}\sqrt{rn}+nh_2(r/n)+C\sqrt n}
  \left(1+\sum_{s<s_*}\binom ns2^s\right).
  \label{eq:many-large-tail-bound}
\end{equation}

For fixed $0<c_*\le1/2$, the standard entropy bound gives
\[
  \sum_{s<s_*}\binom ns2^s
  \le
  \exp\{n(h_2(c_*)+c_*\log2)+o(n)\}.
\]
All constants and the $o(n)$ term in the two preceding estimates are
uniform over matrices satisfying the hypotheses.  Since
$r/n\to c_*/2$ and $(s_*-r)/n\to c_*/2$, put
\[
  \Psi(c)=-a_{\mathrm{ld}}\sqrt{c/2}
          +h_2(c/2)+h_2(c)+c\log2.
\]
Combining \cref{eq:few-large-tail-bound,eq:many-large-tail-bound} gives,
uniformly over every admissible matrix sequence,
\[
  \limsup_{n\to\infty}\frac1n
  \log|\E_x\mathcal T_{s_*}(\phi)|
  \le\max\left\{\frac{c_*}{2}\log(3\chi),\ \Psi(c_*)\right\}.
\]
The first quantity is negative because $3\chi<1$.  The positive terms in
the second quantity are $O(c_*\log(1/c_*))$, which is
$o(\sqrt{c_*})$ as $c_*\downarrow0$.  We may therefore choose
$0<c_*\le\min\{c_{\max},1/2\}$ so that both quantities are strictly
negative.  Taking half of the smaller resulting margin and increasing
$n_{\mathrm{tail}}$ proves
\[
  |\E_x\mathcal T_{s_*}(\phi)|\le e^{-a_{\mathrm{tail}}n}.
\]

Finally, expand $\mathcal T_{s_*}$ into its finite sum.  Applying
  \cref{lem:gaussian-marginalization} to each term shows that
\[
  \E_x\mathcal T_{s_*}(\phi)
  =\sum_{|S|\ge s_*}W_K(S).
\]
This is \eqref{eq:linear-cut-signed-tail}.
\end{proof}

\subsection{From tail bounds to an efficient cutoff}

\begin{proof}[Proof of \cref{thm:gaussian-product-truncation}]
If $\beta=0$, then the entrywise hypothesis forces $K=0$.  Consequently,
$\mathcal Z_K=1$ and every nonempty $W_K(S)$ vanishes, so all conclusions
are immediate.  Assume from now on that $\beta>0$.

By \cref{lem:partial-product-envelope} with $T=[n]$, the expectation defining
$\mathcal Z_K$ converges absolutely.  Expanding $g=1+f$ gives a finite
pointwise sum over coordinate subsets.  For each set $S$,
\cref{rem:factor-consequences} gives the envelope required by
\cref{lem:gaussian-marginalization} for
$\prod_{i\in S}f(\phi_i)$.  Marginalizing every term proves
\eqref{eq:gaussian-subset-expansion}, as well as absolute convergence of
every $W_K(S)$.

Put $C_2=eC_{\mathrm{loc}}^{\ell/2}$.
For $2\le s\le c_{\mathrm{loc}}n$, \cref{lem:small-subset-general} and
$\binom ns\le(en/s)^s$ give
\begin{equation}
  \sum_{|S|=s}|W_K(S)|
  \le \Lambda_s,
  \qquad
  \Lambda_s=\left(C_2\left(\frac sn\right)^\alpha\right)^s.
  \label{eq:small-support-layer-bound}
\end{equation}
Moreover,
\[
  \frac{\Lambda_{s+1}}{\Lambda_s}
  =C_2\left(\frac{s+1}{n}\right)^\alpha
    \left(1+\frac1s\right)^{\alpha s}
  \le C_2e^\alpha\left(\frac{s+1}{n}\right)^\alpha.
\]
Choose $c_{\max}>0$ so small that
\begin{equation}
  c_{\max}\le c_{\mathrm{loc}},
  \qquad
  \beta c_{\max}\le\frac18,
  \qquad
  C_2e^\alpha(2c_{\max})^\alpha\le\frac12.
  \label{eq:layer-ratio-choice}
\end{equation}
After increasing $n_0$, we may assume $c_{\max}n\ge1$.  If
$s<c_{\max}n$, then $(s+1)/n\le2c_{\max}$, so the layer bounds decrease
by a factor of at least two throughout this range.

Apply \cref{lem:linear-cut-tail} with this $c_{\max}$, and set
$c_{\mathrm{tr}}=c_*$.  Increase $n_0$ to include the threshold in that
lemma.  Let $2\le R\le c_{\mathrm{tr}}n$, and put
$s_*=\lceil c_{\mathrm{tr}}n\rceil$.  The omitted subsets split into the
layers $R\le s<s_*$ and the signed tail beginning at $s_*$.  By
\cref{eq:small-support-layer-bound,eq:layer-ratio-choice}, the absolute values
of the intermediate layers form a geometric series bounded by
$2\Lambda_R$.
The linear-cut lemma controls the remaining tail.  Hence
\[
  \left|\mathcal Z_K-\sum_{|S|<R}W_K(S)\right|
  \le
  2\left(C_2\left(\frac Rn\right)^\alpha\right)^R
  +e^{-a_{\mathrm{tail}}n}.
\]
This proves \eqref{eq:diffuse-gaussian-truncation-bound} with
$C_1=\sqrt2\,C_2$.  We choose the separate constant $C_0$ for the near-one
estimate below.

For \eqref{eq:diffuse-gaussian-near-one}, the singleton bound gives
$\sum_{i=1}^n|W_K(\{i\})|\le B(16a_{\max}\beta)^2/n$.
The layers beginning at size two contribute at most
$2\Lambda_2=O(n^{-2\alpha})=O(n^{-1})$, and the signed linear tail is
exponentially small.  Enlarging $C_0$ proves
\eqref{eq:diffuse-gaussian-near-one}.  Finally,
$c_{\mathrm{tr}}\le c_{\max}\le1/(8\beta)$, and
$\|K[S]\|_{\mathrm{row}}\le\beta|S|/n\le1/8$ when
$|S|\le c_{\mathrm{tr}}n$.
\end{proof}

The truncation bound becomes algorithmic because its cutoff can be chosen
without paying $n^{O(\log(1/\eps))}$ time.  We record the elementary
calculus fact once.

\begin{lemma}[Least truncation cutoff]
\label{lem:generic-cutoff}
Fix $\alpha\ge1/2$, $C_2>0$, and $c_{\mathrm{tr}}>0$.  There is a constant
$\eta>0$, depending only on $(\alpha,C_2,c_{\mathrm{tr}})$, with the following
property.  Let $b\ge1$ and suppose that $b<\eta n$.  The least integer
$R\ge2$ satisfying
\begin{equation}
  \left(C_2\left(\frac Rn\right)^\alpha\right)^R\le2^{-b}
  \label{eq:least-cutoff-condition}
\end{equation}
exists, satisfies $R\le c_{\mathrm{tr}}n$, and obeys
\begin{equation}
  R\log\frac{en}{R}=O(b+\log n).
  \label{eq:least-cutoff-cost}
\end{equation}
The hidden constant depends only on $(\alpha,C_2,c_{\mathrm{tr}})$.
\end{lemma}

Writing $D=C_2^{1/\alpha}$, the cutoff condition is governed by
$F(x)=x\log(n/(Dx))$.  Monotonicity and minimality give
$F(R)=O(b+\log n)$ and hence \eqref{eq:least-cutoff-cost}.  The elementary
details are given in \cref{subsec:cutoff-proof}.

\section{Perfect matchings in dense graphs}
\label{sec:matching}

\subsection{The matrix theorem and its consequences}
\label{sec:matching-results}

This section states the matrix-level matching theorem used by the algorithm,
two scaling criteria that reduce weighted matching counts to it, and
representative consequences.  Its proof combines the common truncation
principle of \cref{sec:gaussian-principle} with the matching-specific
normalization and integral representation developed below.  The longer
deductions from the scaling criteria are deferred to the appendices.

\subsubsection{The matrix theorem}

For a positive integer $k$, let $\mathbf J_k$, $I_k$, and $\1_k$ denote the
all-ones matrix, the identity matrix, and the all-ones column vector of
dimension $k$, respectively.  We omit subscripts when the dimension is
clear.  Throughout, $\|\cdot\|$ and $\|\cdot\|_{\mathrm F}$ denote the
spectral operator and Frobenius norms.

For even $n$, define the stochastic complete-graph reference matrix
\[
  P_0=\frac{\mathbf J_n-I_n}{n-1}.
\]
In the balanced bipartite case, write $n=2m$ and define
\[
  P_{\mathrm b}
  =
  \begin{pmatrix}
  0&\mathbf J_m/m\\
  \mathbf J_m/m&0
  \end{pmatrix}.
\]
Both reference matrices are symmetric and stochastic with zero diagonal.
The following theorem says that a diffuse centered perturbation of either
matrix admits a deterministic relative approximation.

\begin{theorem}[Hafnian approximation for centered perturbations]
\label{thm:matrix}
Fix $0\le\beta<\infty$ and $\kappa>0$.  For either reference matrix $P$
above, let $n$ be even, let $0<\eps<1$, and let $E$ be a real symmetric
$n$ by $n$ matrix.  Suppose that $P+E$ is entrywise nonnegative and
\begin{equation}
  E\1=0,\qquad E_{ii}=0,\qquad
  \max_{i,j}|E_{ij}|\le\frac\beta n,
  \qquad \|E\|\le1-\kappa.
  \label{eq:matrix-hypotheses}
\end{equation}
In the complete-bipartite case we additionally require $n=2m$ and
\[
  E=\begin{pmatrix}0&D\\D^{\mathsf T}&0\end{pmatrix}
\]
for an $m$ by $m$ matrix $D$.

There is a deterministic real-arithmetic algorithm that
returns a number $\widehat H\ge0$ satisfying
\[
  e^{-\eps}\haf(P+E)\le\widehat H
  \le e^\eps\haf(P+E).
\]
It uses
\[
  n^{O_{\beta,\kappa}(1)}
  \eps^{-O_{\beta,\kappa}(1)}
\]
arithmetic and elementary-function operations.
The implicit constants depend only on $\beta$ and $\kappa$.
\end{theorem}

In the bipartite case, the block condition ensures that every matching
selected from $E$ is bipartite.  The coefficient contributed by the
reference matrix after this selection depends only on the matching size.
This is the additional
structural input needed for the inverse-Gamma identity.  In fact,
\cref{eq:positive-middle} proves that $\haf(P+E)>0$ for all sufficiently
large $n$ under these conditions.  For such $n$, entrywise nonnegativity
may be omitted: the same procedure returns $\widehat H$ satisfying
$|\widehat H-\haf(P+E)|\le(e^\eps-1)|\haf(P+E)|$.
We retain nonnegativity in the all-dimensions statement so that
multiplicative approximation has its usual meaning in the finitely many
smaller dimensions handled by exact evaluation.

\begin{remark}[Doubly stochastic inputs]
The argument of \cref{thm:matrix}, with the finite-precision analysis in
\cref{app:bit-complexity}, gives a deterministic FPTAS for
$\operatorname{per}(Y)$ when $Y$ is a rational doubly stochastic $m$ by
$m$ matrix satisfying $\max_{i,j}Y_{ij}\le C/m$ and
$\sigma_2(Y)\le1-\kappa$, where $\sigma_2(Y)$ is the second-largest
singular value of $Y$.  Its bit complexity is
$m^{O_{C,\kappa}(1)}\eps^{-O_{C,\kappa}(1)}\poly(L_Y+L_\eps)$.
Indeed, apply the theorem to the symmetric dilation of
$Y-\mathbf J_m/m$.
The entrywise hypothesis is independent of the spectral one: for fixed
$0<a<1$, the matrix $aI+(1-a)\mathbf J_m/m$ has second singular value $a$ but
diagonal entries of order one.  On the diffuse class above, the conclusion
strengthens McCullagh's prescribed $O(1/m)$ relative error to arbitrary
requested accuracy \cite{McCullagh2014}.
\end{remark}

\subsubsection{Entropy-scaling criteria}

For a graph $G$, let $A_G$ and $D_G$ be its adjacency and degree matrices,
and put $L_G=D_G-A_G$ and $Q_G=D_G+A_G$.  Write $\lambda_2(L_G)$ for the
second-smallest eigenvalue of $L_G$.  When $\mathcal D(G)$ contains a point
positive on every edge, let $X$ denote the matrix formed from the maximizer
in the definition of $h_A(G)$.  The normalization in
\cref{prop:entropy-normalization} shows that
$X_{uv}=a_{uv}r_ur_v$ on the support.

A nonnegative square matrix has \emph{total support} if its bipartite
support has a perfect matching and every positive entry belongs to one.
For such a matrix, let $Y^*_{ij}=b_{ij}r_ic_j$ denote its weighted Sinkhorn
scaling, with the gauge $\prod_i r_i=\prod_j c_j$.  We write
$\sigma_2(Y^*)$ for its second-largest singular value.

The symmetric criterion below assumes only that $\mathcal D(G)$ contains a
point positive on every edge.  This condition holds whenever every edge of
$G$ belongs to a perfect matching.
To apply the matrix theorem to an input $A$, it remains to know when entropy
scaling produces a perturbation satisfying its entrywise and spectral
hypotheses.  The next criterion expresses this in terms of the scaling
factors and the Laplacian and signless Laplacian of the support.
Both criteria are promise results: the algorithms are guaranteed on inputs
satisfying the displayed hypotheses and need not verify those hypotheses.

\begin{theorem}[Symmetric scaling criterion]
\label{thm:scaling-criterion}
Fix $0<\theta\le1$ and positive constants $c_{\mathrm{lo}},c_{\mathrm{hi}},q',q$.  Let $A$ be a
rational symmetric $n$ by $n$ matrix, where $n$ is even, with zero
diagonal and entries in $\{0\}\cup[\theta,1]$.  Suppose that its support
graph $G$ has the property that $\mathcal D(G)$ contains a point
positive on every edge.
Write the maximum-entropy scaling as
$X_{uv}=a_{uv}r_ur_v$ on the edges, and suppose that
\[
  \frac{c_{\mathrm{lo}}}{\sqrt n}\le r_v\le\frac{c_{\mathrm{hi}}}{\sqrt n}
  \quad(v\in V(G)),
  \qquad
  \lambda_2(L_G)\ge q'n,
  \qquad
  Q_G\succeq qnI.
\]
Then, given $0<\eps<1$, a deterministic algorithm returns
$\widehat H\ge0$ satisfying
\[
  e^{-\eps}\haf(A)\le\widehat H\le e^\eps\haf(A).
\]
If $G$ has a perfect matching, then $\widehat H>0$.
Its bit complexity is
\[
  n^{O_{\theta,c_{\mathrm{lo}},c_{\mathrm{hi}},q',q}(1)}
  \eps^{-O_{\theta,c_{\mathrm{lo}},c_{\mathrm{hi}},q',q}(1)}
  \poly(L_A+L_\eps).
\]
\end{theorem}

For all sufficiently large $n$, the hypotheses of
\cref{thm:scaling-criterion} themselves force the support to contain a
perfect matching.  See \cref{cor:spectral-existence}.

The bipartite analogue has one fewer spectral hypothesis.  The
singular-value gap supplies both the operator-norm bound needed by the
matrix theorem and the strong convexity needed to compute the Sinkhorn
scaling.

\begin{theorem}[Bipartite scaling criterion]
\label{thm:bipartite-scaling-criterion}
Fix $0<\theta\le1$, constants $0<c_{\mathrm{lo}}\le c_{\mathrm{hi}}<\infty$, and
$0<\kappa_{\mathrm b}\le1$.  Let $B$ be a rational $m$ by $m$ matrix with
entries in $\{0\}\cup[\theta,1]$ and total support.  Let
$Y^*_{ij}=b_{ij}r_ic_j$ be its Sinkhorn scaling, and choose the factors in
the gauge
\(\prod_i r_i=\prod_j c_j\).  Suppose that
\[
  \frac{c_{\mathrm{lo}}}{\sqrt m}\le r_i,c_j\le\frac{c_{\mathrm{hi}}}{\sqrt m},
  \qquad
  \sigma_2(Y^*)\le1-\kappa_{\mathrm b}.
\]
Then, given $0<\eps<1$, a deterministic algorithm returns
$\widehat Z_B>0$ satisfying
\[
  e^{-\eps}\operatorname{per}(B)\le\widehat Z_B
  \le e^\eps\operatorname{per}(B).
\]
Its bit complexity is
\[
  m^{O_{\theta,c_{\mathrm{lo}},c_{\mathrm{hi}},\kappa_{\mathrm b}}(1)}
  \eps^{-O_{\theta,c_{\mathrm{lo}},c_{\mathrm{hi}},\kappa_{\mathrm b}}(1)}
  \poly(L_B+L_\eps).
\]
\end{theorem}

The symmetric criterion is proved in
\cref{sec:matching-normalization}, after the normalization has been
constructed.  The bipartite criterion is proved in
\cref{app:bipartite-scaling}.

\subsubsection{Consequences of the scaling criteria}

We first record the regular cases, where the stochastic normalization is
already visible in the input.  The irregular consequences that follow
illustrate the role of entropy scaling.  Proofs of the corollaries in this
subsection are given in \cref{app:scaling-consequences} and
\cref{app:bipartite-scaling}.

\begin{corollary}[Dense regular spectral expanders]
\label{cor:dense-regular-expanders}
Fix $0<c<1$ and $0<\kappa<1$.
\begin{enumerate}
\item[\rm (a)]
Let $G$ be a $d$-regular graph on an even number $n$ of vertices, where
$d\ge cn$, and suppose that
$\|A_Gz\|_2\le(1-\kappa)d\|z\|_2$ for every $z\perp\1$.
There is a deterministic FPTAS for $\#\PM(G)$, with running time
$n^{O_{c,\kappa}(1)}\eps^{-O_{c,\kappa}(1)}$.

\item[\rm (b)]
Let $B$ be the zero-one biadjacency matrix of a $d$-regular bipartite graph
with $m$ vertices on each side, where $d\ge cm$, and suppose that
$\sigma_2(B/d)\le1-\kappa$.  There is a deterministic FPTAS for
$\operatorname{per}(B)$, with running time
$m^{O_{c,\kappa}(1)}\eps^{-O_{c,\kappa}(1)}$.
\end{enumerate}
\end{corollary}

In both parts the input is already scaled.  Ebrahimnejad, Nagda, and Oveis
Gharan obtain randomized approximate counting and sampling in polynomial
time for regular strong expanders
\cite{EbrahimnejadNagdaOveisGharan2022}.  Gamarnik and Katz obtain a
deterministic $(1+\eta)^m$-factor approximation for constant-degree
bipartite expanders \cite{GamarnikKatz2010}.  The corollary above is
deterministic and gives arbitrary relative accuracy, at the price of a
degree linear in the order and a two-sided spectral gap.

Whether every edge of a graph belongs to a perfect matching can be checked
in polynomial time: for each edge $uv$, one tests whether $G-\{u,v\}$ has
a perfect matching using Edmonds' algorithm \cite{Edmonds1965}.  The next
two consequences are the main applications of entropy scaling beyond the
fixed-margin Dirac class.  They handle irregular supports whose density may
lie below one half.

\begin{corollary}[Dense supports with degree and codegree bounds]
\label{cor:pseudorandom-supports}
Fix constants $p>0$, $0\le\rho<1$, $\eta>0$, and $q>0$ such that
$\eta>2\rho p$.  There is a deterministic FPTAS for $\#\PM(G)$ on every
even-order graph $G$ in which every edge belongs to a perfect matching and
which satisfies
the following conditions, where $d_G(v)$ denotes the degree of a vertex
$v$:
\[
\begin{aligned}
  (1-\rho)pn&\le d_G(v)\le(1+\rho)pn
    &&\text{for every vertex $v$},\\
  |N(u)\cap N(v)|&\ge\eta n
    &&\text{for all distinct $u,v$},\\
  Q_G&\succeq qnI,
\end{aligned}
\]
where $Q_G$ is the signless Laplacian.  Its running time is
$n^{O_{p,\rho,\eta,q}(1)}
\eps^{-O_{p,\rho,\eta,q}(1)}$.
For every fixed $p\in(0,1)$, suitable constants $\rho,\eta,q>0$ make these
hypotheses hold with probability tending to one for the Erd\H{o}s--R\'enyi
graph $G(n,p)$ along even values of $n$.
\end{corollary}

The condition $Q_G\succeq qnI$ is a quantitative nonbipartiteness
hypothesis.  The common-neighbor condition makes $G$ connected.  For a
connected graph, $Q_G$ is singular exactly when $G$ is bipartite.  This is
why part~{\rm (b)} of
\cref{thm:main} uses the separate reference matrix $P_{\mathrm b}$ rather
than following from \cref{cor:pseudorandom-supports}.

The bipartite formulation factors out this unavoidable $-1$ mode and
therefore requires no additional signless-Laplacian hypothesis.  The
following consequence reaches zero-one matrices at every fixed positive
density.

\begin{corollary}[Bipartite supports with degree and codegree bounds]
\label{cor:pseudorandom-bipartite}
Fix $0<\theta\le1$, $p>0$, $0\le\rho<1$, and $\eta>0$.  Put
\(d_-=(1-\rho)p, \qquad d_+=(1+\rho)p\),
and suppose that
\begin{equation}
  \eta>d_+-\theta d_-.
  \label{eq:bipartite-pseudorandom-condition}
\end{equation}
Let $B$ be a rational $m$ by $m$ matrix with entries in
$\{0\}\cup[\theta,1]$ and total support.  Suppose that every row and
column of its support has degree between $d_-m$ and $d_+m$, and that every
two distinct rows, and every two distinct columns, have at least $\eta m$
common neighbors.  Then there is a deterministic FPTAS for
$\operatorname{per}(B)$, with bit complexity
\[
  m^{O_{p,\rho,\eta,\theta}(1)}
  \eps^{-O_{p,\rho,\eta,\theta}(1)}
  \poly(L_B+L_\eps).
\]
In particular, if $B$ is a zero-one matrix with independent
$\operatorname{Bernoulli}(p)$ entries for any fixed $p\in(0,1]$, then its
support satisfies these hypotheses with probability tending to one as
$m\to\infty$.
\end{corollary}

The proof, including the random-support assertion, is given in
\cref{app:bipartite-scaling}.

\subsection{Maximum-entropy normalization}
\label{sec:matching-normalization}

We now turn from the range of applications to the proof.  This section
factors the leading entropy out of each input and verifies the centered,
entrywise, and spectral hypotheses of \cref{thm:matrix}.  The symmetric and
bipartite constructions differ only at this normalization stage.  After the
symmetric dilation below, the original support graphs and weight matrices
no longer enter the analytic argument.

\subsubsection{The symmetric normalization}

The first step removes the leading weighted entropy of the input matrix.
For an adjacency matrix, this is the entropy parameter used by Cuckler and
Kahn to determine the exponential scale of the count
\cite{CucklerKahn2009,CucklerKahnHamiltonian2009}.  Here it also supplies
an exact diagonal normalization.  Let $A=(a_{uv})$ satisfy part~{\rm (a)}
of \cref{thm:main}, let $G$ be its support graph, and recall the set
$\mathcal D(G)$ defined in the introduction.
Only the degree equations appear in $\mathcal D(G)$.  In particular, this
is not the perfect-matching polytope with its odd-cut inequalities.
We assume in this section that $n\ge1/\gamma+2$.  The finitely many smaller
inputs are assigned to the exact branch.

\begin{lemma}[Every edge extends to a perfect matching]
\label{lem:edge-extendability}
Suppose that \cref{eq:main-degree} holds and that $n\ge1/\gamma+2$.
Every edge of $G$ belongs to a perfect matching.
\end{lemma}

\begin{proof}
Fix an edge $uv$.  To extend $uv$, it is enough to find a perfect matching
after deleting its endpoints.  The remaining graph has $n-2$
vertices and minimum degree at least
\(\left(\frac12+\gamma\right)n-2 \ge \frac{n-2}{2}\).
Dirac's theorem gives a Hamiltonian cycle in this graph.  Alternating edges
of the cycle form a perfect matching, which extends to one of $G$ after
adding $uv$.
\end{proof}

For a vector $w\in\mathcal D(G)$, define its weighted entropy by
\(F_A(w)=\sum_{uv\in E(G)}w_{uv}\log\frac{a_{uv}}{w_{uv}}\),
where a summand with $w_{uv}=0$ is interpreted as zero.  The set
$\mathcal D(G)$ is nonempty because $G$ has a perfect matching.  Let $w^*$
maximize $F_A$, so $F_A(w^*)=h_A(G)$.
By \cref{lem:edge-extendability}, averaging one perfect matching through
each edge gives a point of $\mathcal D(G)$ that is positive on every edge.

\begin{proposition}[Exact entropy normalization]
\label{prop:entropy-normalization}
Let $A$ be any nonnegative symmetric matrix with zero
diagonal and support graph $G$.  Suppose that $\mathcal D(G)$ contains a
point positive on every edge, and let $w^*$ maximize $F_A$ on
$\mathcal D(G)$.  Then $w^*$ is the unique maximizer and is positive on
every edge.
There are positive numbers $r_v$, indexed by the vertices of $G$, such that
\[
  w^*_{uv}=a_{uv}r_ur_v
  \qquad(uv\in E(G)).
\]
Let $X$ be the symmetric matrix with entries $X_{uv}=w^*_{uv}$ on edges of
$G$, and zero entries on nonedges and on the diagonal.  Then
\begin{equation}
  \haf(A)=e^{h_A(G)}\haf(X).
  \label{eq:entropy-normalization}
\end{equation}
\end{proposition}

\begin{proof}
For every $a>0$, the function $x\mapsto x\log(a/x)$, extended continuously
by zero at $x=0$, is strictly concave on $[0,\infty)$.  Thus $F_A$ is
strictly concave on the convex set $\mathcal D(G)$, and its maximizer is
unique.

Let $\bar w\in\mathcal D(G)$ be positive on every edge, and suppose for a
contradiction that $w^*$ has a zero coordinate.  For
$w(t)=(1-t)w^*+t\bar w$,
\[
  F_A(w(t))-F_A(w^*)
  =t\log(1/t)\sum_{e:w_e^*=0}\bar w_e+O(t)>0
\]
for all sufficiently small $t$.  The displayed leading term comes from
the zero coordinates, while the positive coordinates change by $O(t)$.
This contradiction proves that $w^*$ is positive on every edge.

Let $d_v$ be the multiplier for the degree equation at $v$.
Stationarity gives
$\log(w^*_{uv}/a_{uv})=d_u+d_v-1$.  Hence, with
$r_v=e^{d_v-1/2}$, one has $w^*_{uv}=a_{uv}r_ur_v$.  Moreover,
\[
  \sum_v \log r_v
  =\sum_{uv\in E(G)}w^*_{uv}(\log r_u+\log r_v)
  =\sum_{uv\in E(G)}w^*_{uv}\log\frac{w^*_{uv}}{a_{uv}}
  =-h_A(G).
\]
Every perfect matching $M$ uses each factor $r_v$ once, so
$\prod_{uv\in M}a_{uv}=(\prod_vr_v)^{-1}\prod_{uv\in M}X_{uv}$.
Summing over $M$ gives
$\haf(A)=(\prod_v r_v)^{-1}\haf(X)=e^{h_A(G)}\haf(X)$, which is
\cref{eq:entropy-normalization}.
\end{proof}

The degree margin forces every scaling factor to have order $n^{-1/2}$ and
every nonzero entry of $X$ to have order $n^{-1}$.  Put
$\alpha=1/2+\gamma$.  For vertex sets $S,T$, write
$X(S,T)=\sum_{u\in S,v\in T}X_{uv}$, and put
$r(S)=\sum_{v\in S}r_v$.

\begin{lemma}[Uniform bounds for the symmetric scaling]
\label{lem:diffuse-scaling}
The scaling factors satisfy
\begin{equation}
  \theta\sqrt{\frac{2\gamma}{n}}
  \le r_v\le
  \frac1{\theta\sqrt{2\gamma n}}
  \qquad(v\in V(G)).
  \label{eq:graph-factor-bounds}
\end{equation}
Consequently, every edge $uv$ satisfies
\begin{equation}
  \frac{2\gamma\theta^3}{n}
  \le X_{uv}\le
  \frac1{2\gamma\theta^2n}.
  \label{eq:diffuse-X}
\end{equation}
\end{lemma}

\begin{proof}
We first show that the vertices with smaller scaling factors carry
substantial total mass.  Every neighborhood must meet that mass, which
gives the upper factor bound.  The row equations then give the lower one.

Write $N(v)$ for the neighborhood of a vertex $v$.
Let $V_{\mathrm{high}}$ consist of the $\lfloor(1-\alpha)n\rfloor$
vertices with largest $r$-values, and put
$V_{\mathrm{low}}=V(G)\setminus V_{\mathrm{high}}$.  Symmetry and
stochasticity give
$X(V_{\mathrm{low}},V_{\mathrm{low}})
\ge |V_{\mathrm{low}}|-|V_{\mathrm{high}}|\ge2\gamma n$.
Indeed, the mass from $V_{\mathrm{low}}$ to $V_{\mathrm{high}}$ is at most
the total column mass $|V_{\mathrm{high}}|$.
On the other hand, $a_{uv}\le1$ gives
$X(V_{\mathrm{low}},V_{\mathrm{low}})\le r(V_{\mathrm{low}})^2$.  Hence
$r(V_{\mathrm{low}})\ge\sqrt{2\gamma n}$.

For every vertex $v$, the set $V(G)\setminus N(v)$ has at most
$|V_{\mathrm{high}}|$ vertices.  Since $V_{\mathrm{high}}$ contains the
largest factors,
\(\sum_{u\in N(v)}r_u =r(V(G))-r(V(G)\setminus N(v)) \ge r(V(G))-r(V_{\mathrm{high}})=r(V_{\mathrm{low}})\).
The row equation and $a_{uv}\ge\theta$ on the support now give
\(1=r_v\sum_{u\in N(v)}a_{uv}r_u \ge\theta r_v\sqrt{2\gamma n}\),
so $r_v\le(\theta\sqrt{2\gamma n})^{-1}$.  Conversely, the same row
equation and $a_{uv}\le1$ give, with
$r_{\max}=\max_{u\in V(G)}r_u$,
\(1\le r_v(n-1)r_{\max}\),
and therefore
\(r_v\ge\frac{\theta\sqrt{2\gamma n}}{n-1} \ge\theta\sqrt{\frac{2\gamma}{n}}\).
Since $X_{uv}=a_{uv}r_ur_v$ on an edge, these two factor bounds prove
\cref{eq:diffuse-X}.
\end{proof}

We next separate the constant mode.  Recall the complete-graph reference
matrix $P_0$ from \cref{sec:matching-results}, and put $E=X-P_0$.  Both $X$
and $P_0$ are stochastic, so $E\1=0$.  They also have zero diagonal.

To control $E$ spectrally, we use the Laplacian $L_G$ and signless
Laplacian $Q_G$ introduced before the scaling criteria.  The Laplacian
controls fluctuations orthogonal to $\1$, while the signless Laplacian
rules out an almost-bipartite negative mode.

\begin{lemma}[Spectral bounds for the support graph]
\label{lem:centered-gap}
Suppose that \cref{eq:main-degree} holds.  Then
\[
  \lambda_2(L_G)\ge2\gamma n,
  \qquad
  Q_G\succeq2\gamma nI.
\]
\end{lemma}

\begin{proof}
Let $H=\overline G$ be the complement of $G$, and write $\Delta(H)$ for
its maximum degree.  We have
\(\Delta(H)\le(1/2-\gamma)n-1\).
On $\1^\perp$, the identity $L_G=L_{K_n}-L_H$ and the bound
$\lambda_{\max}(L_H)\le2\Delta(H)$ give
$\lambda_2(L_G)\ge n-2\Delta(H)\ge2\gamma n$.
Likewise,
$Q_G+Q_H=(n-2)I+\mathbf J$ and
$\lambda_{\max}(Q_H)\le2\Delta(H)$.  Hence
$Q_G\succeq((n-2)-2\Delta(H))I\succeq2\gamma nI$.
\end{proof}

\subsubsection{The bipartite normalization}

The permanent case uses the same analytic theorem after a different
normalization.  Weighted Sinkhorn scaling first removes the entropy, and a
symmetric dilation then represents the permanent of the scaled matrix as a
hafnian.  We state the estimates needed for this reduction here and defer
their elementary but longer proofs to \cref{app:bipartite-scaling}.

Let $B=(b_{ij})$ be a nonnegative $m$ by $m$ matrix.  Define
$\mathcal D_{\mathrm{bi}}(B)$ to be the polytope of doubly stochastic
matrices whose nonzero entries occur only where $B$ is positive:
\[
  \mathcal D_{\mathrm{bi}}(B)
  =
  \left\{Y\in\R_{\ge0}^{m\times m}:
  Y_{ij}=0\text{ if }b_{ij}=0,\quad
  Y\1=\1,\quad Y^{\mathsf T}\1=\1\right\}.
\]
For $Y\in\mathcal D_{\mathrm{bi}}(B)$, define
$F_B(Y)=\sum_{i,j:b_{ij}>0}Y_{ij}\log(b_{ij}/Y_{ij})$, with a zero
summand when $Y_{ij}=0$.  When $\mathcal D_{\mathrm{bi}}(B)$ is nonempty,
write
$h_B=\max_{Y\in\mathcal D_{\mathrm{bi}}(B)}F_B(Y)$.

Recall that total support means that the bipartite support has a perfect
matching and every positive entry belongs to one.

\begin{proposition}[Bipartite entropy normalization]
\label{prop:bipartite-normalization}
Suppose that $B$ has total support.  Then the weighted entropy maximizer
$Y^*$ in $\mathcal D_{\mathrm{bi}}(B)$ is positive at every positive entry
and has the form
\(Y^*_{ij}=b_{ij}r_ic_j\qquad(b_{ij}>0)\)
for positive row factors $r_i$ and column factors $c_j$.  Moreover,
\begin{equation}
  \operatorname{per}(B)=e^{h_B}\operatorname{per}(Y^*).
  \label{eq:bipartite-entropy-normalization}
\end{equation}
\end{proposition}

The factorization follows from the same relative-interior and
Lagrange-multiplier argument used in
\cref{prop:entropy-normalization}, with separate multipliers for rows and
columns.  Existence and uniqueness of the doubly stochastic scaled matrix
under total support are due to Sinkhorn and Knopp
\cite{SinkhornKnopp1967}.  See \cref{app:bipartite-scaling}.  The fixed
degree margin gives the additional
quantitative estimates needed below.

We now impose the degree-margin hypothesis in part~{\rm (b)} of
\cref{thm:main}.  The factors are unchanged if all $r_i$ are multiplied
by one positive constant and all $c_j$ are divided by the same constant.
We remove this ambiguity by requiring
$\prod_i r_i=\prod_j c_j$.  Equivalently, there are vectors $x^*$ and
$y^*$ with $\sum_i x_i^*=\sum_j y_j^*$ such that
$r_i=m^{-1/2}e^{x_i^*}$ and $c_j=m^{-1/2}e^{y_j^*}$.

\begin{proposition}[Uniform bounds for the bipartite scaling]
\label{prop:bipartite-scaling}
Suppose that every row and column of the support of $B$ contains at least
$(1/2+\gamma)m$ entries, that its nonzero entries lie in $[\theta,1]$,
and that $m\ge\lceil(2\gamma)^{-1}\rceil$.  Then $B$ has total support.
Every positive entry satisfies
\begin{equation}
  \frac{2\gamma\theta^3}{m}
  \le Y^*_{ij}\le
  \frac1{2\gamma\theta^2m}.
  \label{eq:bipartite-diffuse}
\end{equation}
Under this normalization, there are positive constants
$c_{\gamma,\theta}$ and $C_{\gamma,\theta}$ such that the individual
factors satisfy
\begin{equation}
  \frac{c_{\gamma,\theta}}{\sqrt m}
  \le r_i,c_j\le
  \frac{C_{\gamma,\theta}}{\sqrt m}.
  \label{eq:bipartite-factor-box}
\end{equation}
Finally, let $\sigma_2(Y^*)$ denote the second-largest singular value of
$Y^*$.  Then
\begin{equation}
  \sigma_2(Y^*)\le1-4\gamma^3\theta^6.
  \label{eq:bipartite-gap}
\end{equation}
\end{proposition}

The proof combines Hall's theorem, a mass comparison for the row and
column factors, and a common-neighbor estimate for the second singular
value.  See \cref{app:bipartite-scaling}.  Only the entrywise bounds, the
uniform bounds on the factors, and the singular-value gap will be used in
the main argument.

To apply the matrix theorem, represent the permanent as a hafnian by
symmetric dilation.  Recall $P_{\mathrm b}$ from
\cref{sec:matching-results}, and put $D=Y^*-\mathbf J_m/m$ and
\[
  E_{\mathrm b}
  =
  \begin{pmatrix}
  0&D\\
  D^{\mathsf T}&0
  \end{pmatrix}.
\]
Then $E_{\mathrm b}\1=0$, its diagonal blocks vanish, and, in the ambient
dimension $n=2m$, one has
$\max_{i,j}|(E_{\mathrm b})_{ij}|\le(\gamma\theta^2n)^{-1}$ and
$\|E_{\mathrm b}\|\le1-4\gamma^3\theta^6$.
Indeed, the entrywise estimate follows from
\cref{eq:bipartite-diffuse}, while
$\|Y^*-\mathbf J_m/m\|=\sigma_2(Y^*)$ and the symmetric dilation has the same norm.
Moreover, $\haf(P_{\mathrm b}+E_{\mathrm b})=\operatorname{per}(Y^*)$
and $\haf(P_{\mathrm b})=m!/m^m$.

Thus the bounds on the factors control the entries of $E_{\mathrm b}$,
while the singular-value gap controls its operator norm.  The hypotheses of
the matrix theorem follow.  This dilation is the only step in which the
bipartite reduction differs from the symmetric one.

\subsubsection{Proof of the scaling criteria}

The normalizations above give the identities and quantitative bounds used
in the two criteria stated in \cref{sec:matching-results}.  We now verify
their remaining hypotheses.

\begin{proof}[Proof of \cref{thm:scaling-criterion}]
It suffices to verify the hypotheses of the matrix theorem after the exact
entropy normalization.  The bounds on the factors give the entrywise bound,
while
the Laplacian and signless Laplacian keep the spectrum of $E=X-P_0$ away
from $+1$ and $-1$, respectively.

The support hypothesis allows us to apply
\cref{prop:entropy-normalization}.  The factor bounds give
$\theta c_{\mathrm{lo}}^2/n\le X_{uv}\le c_{\mathrm{hi}}^2/n$ for
$uv\in E(G)$.
For $z\perp\1$, comparison with the two unweighted quadratic forms gives
\[
  z^{\mathsf T}(I-X)z
  \ge\frac{\theta c_{\mathrm{lo}}^2}{n}z^{\mathsf T}L_Gz,
  \qquad
  z^{\mathsf T}(I+X)z
  \ge\frac{\theta c_{\mathrm{lo}}^2}{n}z^{\mathsf T}Q_Gz.
\]
Thus the spectrum of $X$ on $\1^\perp$ lies in
\([-1+\theta c_{\mathrm{lo}}^2q,\,1-\theta c_{\mathrm{lo}}^2q']\).
Put $E=X-P_0$.  Choose fixed constants satisfying
$\beta>c_{\mathrm{hi}}^2+2$ and
$0<\kappa<\min\{1/2,(\theta c_{\mathrm{lo}}^2/2)\min\{q,q'\}\}$,
and enlarge a fixed threshold $n_0$ so that
$1/(n-1)\le\kappa$ for $n\ge n_0$.
These choices verify every condition in
\cref{eq:matrix-hypotheses}, as well as $P_0+E=X\ge0$.
Although $E$ need not be rational, \cref{thm:matrix} applies to it in real
arithmetic.  The analysis in \cref{app:bit-complexity} shows that the
scaling and the truncated formula can
be evaluated with polynomially many bits.  Restoring the entropy factor
recovers $\haf(A)$.  The remaining finitely many dimensions
are evaluated exactly.  The exact branch returns zero when no perfect
matching exists.  If one exists, both branches return a positive number.
\end{proof}

The normalization identity and symmetric dilation above reduce the permanent
to \cref{thm:matrix}.  The bounds on the factors supply the entrywise bound
and the singular-value gap supplies the norm bound.  The formal reduction is
recorded in \cref{app:bipartite-scaling}, and its finite-precision
implementation follows from \cref{app:bit-complexity}.

This completes the normalization stage.  In either model, the original
count is an explicit entropy factor times $\haf(P+E)$, where
$P\in\{P_0,P_{\mathrm b}\}$ and $E$ satisfies
\cref{eq:matrix-hypotheses} with constants determined only by the fixed
input parameters.

\subsection{An inverse-Gamma representation}
\label{sec:matching-radial}

The proof of \cref{thm:matrix} begins with an exact one-dimensional
representation.  Fix one of the two reference matrices $P$ from
\cref{sec:matching-results} and a perturbation $E$ satisfying
\cref{eq:matrix-hypotheses}.  Expanding $\haf(P+E)$ according to the edges
chosen from $E$ leaves coefficients that depend only on the number of chosen
edges.  We identify these coefficients with negative moments of one Gamma
random variable, obtaining an average of a signed matching generating
polynomial.

Fix $\beta\ge0$ and $\kappa>0$, and suppose that $E$ satisfies
\cref{eq:matrix-hypotheses}.  If $\kappa=1$, then $E=0$ and the hafnian is
explicit.  If $\kappa>1$, there are no admissible inputs.  We therefore
assume $0<\kappa<1$.  To
treat the two reference matrices in one formula, let $\tau=0$ for the
complete-graph reference $P=P_0$, and let $\tau=1$ for the
complete-bipartite reference $P=P_{\mathrm b}$.  Define
\(a_n=\frac{n+1+\tau}{2}, \qquad b_n=\frac{n-1+\tau}{2}\),
and let $U$ be a Gamma random variable with shape $a_n$ and rate $b_n$.
The variable $U$ therefore has mode $(a_n-1)/b_n=1$ in both cases.  Thus the Gamma mass
is centered at the parameter value where the subsequent subset expansion
is controlled.  For a matching
$M$ on $[n]$, let
$\mathrm{wt}_E(M)$ denote the product of the entries $E_{uv}$ over its
edges.  Define the weighted matching generating polynomial
\[
  Z_E(t)=\sum_{M\text{ a matching}}t^{|M|}\mathrm{wt}_E(M).
\]
We normalize the desired hafnian by its reference value and write
\(Q_E=\frac{\haf(P+E)}{\haf(P)}\).

\begin{lemma}[Inverse-Gamma identity]
\label{lem:radial-identity}
The normalized hafnian has the exact representation
\begin{equation}
  Q_E=\E Z_E(U^{-1}).
  \label{eq:radial-identity}
\end{equation}
For the complete-graph reference,
$\haf(P)=(n-1)!!(n-1)^{-n/2}$.  For the complete-bipartite reference,
where $n=2m$, one has $\haf(P)=m!/m^m$.
\end{lemma}

\begin{proof}
We compare the coefficients of $Z_E$ with the negative
moments of $U$, treating the two reference matrices separately.

For $0\le k\le n/2$, let $m_k(E)$ be the sum of $\mathrm{wt}_E(M)$ over all
$k$-matchings.  In the complete-graph case, expanding according to the
edges on which $E$ is selected gives
\begin{align*}
  \haf(P_0+E)&=\sum_{k=0}^{n/2}
  (n-1)^{-(n/2-k)}(n-2k-1)!!\,m_k(E),\\
  \E U^{-k}
  &=\left(\frac{n-1}{2}\right)^k
   \frac{\Gamma((n+1)/2-k)}{\Gamma((n+1)/2)}
  =(n-1)^k\frac{(n-2k-1)!!}{(n-1)!!}.
\end{align*}
Here and below, $(-1)!!=1$.
This proves the identity for $P_0$.

Now let $n=2m$ and $P=P_{\mathrm b}$.  The block condition on $E$ makes
every matching selected from $E$ bipartite.  A $k$-matching leaves
$m-k$ vertices on each side, which the reference matrix matches in
$(m-k)!$ ways.  Consequently,
\begin{align*}
  \haf(P_{\mathrm b}+E)
  &=\sum_{k=0}^m m^{-(m-k)}(m-k)!\,m_k(E),\\
  \E U^{-k}
  &=m^k\frac{\Gamma(m+1-k)}{\Gamma(m+1)}
  =m^k\frac{(m-k)!}{m!}.
\end{align*}
Here $U$ has shape $m+1$ and rate $m$.
Since $\haf(P_{\mathrm b})=m!/m^m$, substitution into $Z_E$ proves the
bipartite case.
\end{proof}

The inverse-Gamma identity reduces the matrix problem to controlling
$Z_E(t)$ for parameters $t$ near one.  We represent this polynomial by
Wick's rule, using the complex Gaussian convention introduced at the
beginning of \cref{sec:gaussian-principle}.  In particular,
$A_E=E_+^{1/2}+iE_-^{1/2}$ has bilinear second-moment matrix $E$.

Let $\xi=A_Ex$ and define the energy
$\mathcal E_E(x)=x^{\mathsf T}|E|x$.  Wick's rule gives
\[
  Z_E(t)=\E_x\prod_{i=1}^n(1+\sqrt t\,\xi_i)
  \qquad(t\ge0).
\]
Since $E\1=0$, one has $A_E\1=0$ and hence $\sum_i \xi_i=0$.  The squared
moduli of the factors therefore satisfy
\(\sum_{i=1}^n|1+\sqrt t\,\xi_i|^2 =n+t\sum_{i=1}^n|\xi_i|^2 =n+t\,x^{\mathsf T}|E|x\).
Applying AM--GM now gives
\begin{equation}
  |Z_E(t)|\le
  \E_x\left(1+\frac{t\mathcal E_E(x)}{n}\right)^{n/2}.
  \label{eq:radial-amgm}
\end{equation}

The density of $U$ is
\begin{equation}
  p_{n,\tau}(u)=\frac{b_n^{a_n}}{\Gamma(a_n)}
  u^{a_n-1}e^{-b_nu}\qquad(u>0).
  \label{eq:gamma-density}
\end{equation}

For an event $\mathcal A$, the notation $\ind{\mathcal A}$ denotes its
indicator.

\begin{proposition}[Localization of the Gamma average]
\label{prop:radial-localization}
Fix $0\le\beta<\infty$ and $0<\kappa<1$.  Suppose that $E$ satisfies the
four conditions in \cref{eq:matrix-hypotheses}.  Choose fixed numbers
$r_0$ and
$r_1$ with
\(0<1-\kappa<r_0<1<r_1\).
There are constants $c>0$ and $n_0$, depending only on
$\beta,\kappa,r_0,r_1$, such that
\begin{equation}
  \E\left[|Z_E(U^{-1})|
  \ind{U\notin[r_0,r_1]}\right]
  \le e^{-cn}
  \label{eq:radial-tail}
\end{equation}
whenever $n\ge n_0$.
\end{proposition}

The upper and lower tails are controlled by combining
\cref{eq:radial-amgm} with the corresponding Gamma Chernoff bounds.  The
Gaussian factor contributes only $e^{O(\sqrt n)}$, while the Gamma density
has a strictly positive linear rate away from its mode at one.  The details
are given in \cref{subsec:gamma-localization-proof}.

Henceforth we work only on the compact parameter interval
\(\mathcal I=[1/r_1,1/r_0]\).
For every $t\in\mathcal I$,
\(t\|E\|\le\frac{1-\kappa}{r_0}<1\).
This strict margin permits the Gaussian change of measure used next.

\subsection{Matching-specific Gaussian completion}
\label{sec:matching-completion}

The goal of this subsection is to turn every retained parameter value into
an admissible Gaussian product, uniformly over the localization interval,
and then derive a real formula for each retained coefficient.

The Wick representation still contains one factor
$1+\sqrt t\,\xi_i$ at every vertex.  Centering removes their combined
linear term.  We then absorb the quadratic term into the Gaussian measure,
leaving a residual factor whose Taylor expansion begins in degree three.
For $t\in\mathcal I$, define the transformed matrix, the corresponding
single-coordinate factor, and its nonconstant part by
\[
  K_t=tE(I+tE)^{-1},\qquad
  g(z)=(1+z)e^{-z+z^2/2},
  \qquad f(z)=g(z)-1.
\]
The key point is that $f(z)=O(z^3)$ at the origin.

\begin{lemma}[Matching scalar factor]
\label{lem:matching-factor}
The function $g$ is admissible with envelope weights
$a_{\mathrm R}=1$ and $a_{\mathrm I}=0$, vanishing order $\ell=3$,
$C_f=3$, and $\delta_f=1/2$.
\end{lemma}

The elementary proof is given in
Appendix~\ref{subsec:matching-factor-proof}.

\begin{lemma}[Gaussian change of measure]
\label{lem:gaussian-completion}
For every $t\in\mathcal I$,
\begin{equation}
  Z_E(t)=\det(I+tE)^{-1/2}
  \E_{K_t}\prod_{i=1}^ng(\phi_i).
  \label{eq:gaussian-completion}
\end{equation}
The square root is the positive one.
\end{lemma}

\begin{proof}
Put $C=tE$ and let
\(\psi=(C_+^{1/2}+iC_-^{1/2})x\).
Wick's rule gives
\(Z_E(t)=\E_x\prod_{i=1}^n(1+\psi_i)\).
Since $E\1=0$, spectral calculus gives
$C_+^{1/2}\1=C_-^{1/2}\1=0$.  Hence $\sum_i\psi_i=0$.  The positive and
negative spectral parts have disjoint supports, so
$\sum_i\psi_i^2=x^{\mathsf T}Cx$.  The definition of $g$ now gives
\[
  \prod_{i=1}^n(1+\psi_i)
  =e^{-\frac12x^{\mathsf T}Cx}
   \prod_{i=1}^ng(\psi_i).
\]
Since $\|C\|<1$, the matrix $I+C$ is positive definite.  Absorbing the
quadratic factor into the Gaussian density and using
$y=(I+C)^{1/2}x$ contributes $\det(I+C)^{-1/2}$.  With
$K=C(I+C)^{-1}$, spectral calculus gives
$A_C(I+C)^{-1/2}=A_K$, proving \cref{eq:gaussian-completion}.

Every positive eigenvalue of $K$ is smaller than $1/2$.  Therefore
\cref{lem:matching-factor} gives
\[
  \E_y\prod_i|g((A_Ky)_i)|
  \le\E_ye^{y^{\mathsf T}K_+y}
  =\det(I-2K_+)^{-1/2}<\infty.
\]
This justifies the change of variables absolutely.  The square root is the
positive one because $I+C$ is positive definite.
\end{proof}

The expansion over coordinate subsets requires three facts uniformly in $t$: the
transformed complex Gaussian integral retains an integrability margin, the entries
of $K_t$ remain $O(1/n)$, and the determinant prefactor stays of constant
order.  Put $\rho_0=1-\kappa$ and
\(q_0=\frac{\rho_0}{r_0}<1\).
Thus $q_0$ is a uniform upper bound for $t\|E\|$ on $\mathcal I$.

\begin{lemma}[Uniform bounds for $K_t$]
\label{lem:uniform-completion}
Under the hypotheses and notation of \cref{prop:radial-localization}, there
are constants $0<d\le1$, $\beta'<\infty$, and $C<\infty$
that depend only on $\beta,\kappa,r_0,r_1$, such that
\begin{equation}
  \begin{gathered}
  I-2K_t\succeq dI,
  \qquad
  \max_{i,j}|(K_t)_{ij}|\le\frac{\beta'}n,\\[2pt]
  e^{-C}\le\det(I+tE)^{-1/2}\le e^C.
  \end{gathered}
  \label{eq:uniform-K}
\end{equation}
for every $t\in\mathcal I$.
\end{lemma}

\begin{proof}
The spectral calculus, resolvent, and trace estimates are recorded in
\cref{subsec:matching-resolvent-estimates}.
\end{proof}

\Cref{lem:matching-factor} shows that $g$ is admissible with envelope weights
$a_{\mathrm R}=1$ and $a_{\mathrm I}=0$, vanishing order $\ell=3$, and
the fixed local constants $C_f=3$ and $\delta_f=1/2$.
By \cref{lem:uniform-completion}, the matrices $K_t$ satisfy the diffuse
entrywise bound and $I-2K_t\succeq dI$.  Since $d\le1$, spectral calculus gives
$I-2(K_t)_+\succeq dI$.  Thus the matrices $K_t$ are admissible with
fixed $\beta'$ and $d$, uniformly for $t\in\mathcal I$.
The truncation theorem in \cref{thm:gaussian-product-truncation} therefore
applies throughout this interval.  We retain its notation $W_{K_t}(S)$
for the contribution of a coordinate set $S$.

\subsubsection{Evaluation of a retained subset}

The truncation theorem identifies the subsets that the algorithm retains.
We now evaluate the contribution of one such subset without enumerating
Gaussian pairings.  Because $f=g-1$, inclusion--exclusion reduces
$W_K(S)$ to Gaussian products of $g$, and a weighted monomer--dimer
recurrence evaluates each product.

\begin{proposition}[Fixed-subset evaluation]
\label{prop:fixed-support}
Let $K$ be a real symmetric matrix.  Suppose that $S$ has size $s$ and
\(\|K[S]\|_{\mathrm{row}}\le1/4\).
Then $W_K(S)$ can be computed in $3^s\poly(s)$ arithmetic and
elementary-function operations.
\end{proposition}

\begin{proof}
We first evaluate the product of the residual factors on a fixed subset
$T$, using a Gaussian shift and a monomer--dimer recurrence.
Inclusion--exclusion then recovers $W_K(S)$.

For $T\subseteq S$, write
\(G_K(T)=\E_{K[T]}\prod_{i\in T}g(\phi_i), \qquad G_K(\varnothing)=1\).
For nonempty $T$, let $C_T=K[T]$ be the restricted bilinear second-moment
matrix, let $A_T=A_{C_T}$ be the complex square-root factor defined
in \cref{sec:gaussian-principle}, and put
\(D_T=I-C_T,\qquad a_T=A_T^{\mathsf T}\1\).
The matrix $D_T$ is the Gaussian precision matrix left after the quadratic
term in $\prod_{i\in T}g(\phi_i)$ is absorbed.  For every $T\subseteq S$,
$\|K[T]\|_{\mathrm{row}}\le\|K[S]\|_{\mathrm{row}}\le1/4$.
Hence $\|C_T\|\le1/4$ and $D_T\succeq3I/4$.  The definition of $g$ now
gives the absolutely convergent integral
\[
  G_K(T)
  =(2\pi)^{-|T|/2}\int_{\R^T}
  e^{-x^{\mathsf T}D_Tx/2-a_T^{\mathsf T}x}
  \prod_{i\in T}\bigl(1+(A_Tx)_i\bigr)\,dx.
\]
Complete the square and set $y=x+D_T^{-1}a_T$.  The new contour is the
complex translate $\R^T+D_T^{-1}a_T$.  The integrand is entire, and the
positive definiteness of $D_T$ makes the Gaussian dominate the polynomial
on every intermediate contour, so Cauchy's theorem shifts the coordinates
back to $\R^T$.  Put
$\Sigma_T=D_T^{-1}C_T$ and $\mu_T=-\Sigma_T\1$.
Because $A_T$ and $D_T$ are commuting spectral functions of $C_T$, the
shifted affine Gaussian has mean $\mu_T$, bilinear second-moment matrix
$\Sigma_T$, and normalization exponent
$a_T^{\mathsf T}D_T^{-1}a_T=\1^{\mathsf T}\Sigma_T\1$.
Define $P_T$ to be the weighted monomer--dimer partition function on $T$,
with monomer weight $1+(\mu_T)_i$ at $i$ and dimer weight
$(\Sigma_T)_{ij}$ on $\{i,j\}$ \cite{HeilmannLieb1972}.  Explicitly,
$P_T(\varnothing)=1$.  For nonempty $U\subseteq T$, choose a fixed $i\in U$
and use
\[
  P_T(U)=\bigl(1+(\mu_T)_i\bigr)P_T(U\setminus\{i\})
  +\sum_{j\in U\setminus\{i\}}
  (\Sigma_T)_{ij}P_T(U\setminus\{i,j\}).
\]
Wick's formula therefore gives
\begin{equation}
  G_K(T)=\det(D_T)^{-1/2}
  \exp\left(\frac12\1^{\mathsf T}\Sigma_T\1\right)P_T(T).
  \label{eq:fixed-support-formula}
\end{equation}

The row-norm bound gives
$\max_{i,j\in S}|K_{ij}|\le1/4$ and $I-2K[S]_+\succeq I/2$.
Thus, as an $s$ by $s$ moment matrix, $K[S]$ satisfies the hypotheses of
\cref{lem:gaussian-marginalization} with $\beta=s/4$ and $d=1/2$.
Consequently, inclusion--exclusion and
\cref{lem:gaussian-marginalization} give
\[
  W_K(S)=\sum_{T\subseteq S}(-1)^{|S|-|T|}G_K(T).
\]
The recurrence for a fixed $T$ uses $2^{|T|}\poly(s)$ operations, and
\(\sum_{T\subseteq S}2^{|T|}=3^s\).
\end{proof}

Although the derivation uses complex Gaussian contours, the evaluation
formula itself is real.  Indeed, $D_T$ is real positive definite,
$\Sigma_T$ and $\mu_T$ are real, and the monomer--dimer recurrence has real
weights.  Thus the algorithm uses only real arithmetic.  The complex square
roots and contour shifts are analytic devices used to establish the
formula.

\subsection{The approximation algorithm}
\label{sec:matching-algorithm}

The analytic estimates are now in place.  Localization restricts the
Gamma average to a fixed interval, and Gaussian completion followed by
\cref{thm:gaussian-product-truncation} controls the resulting subset
expansion.
We now choose the truncation order and approximate the remaining
one-dimensional integral.

\subsubsection{Truncation, positivity, and evaluation bounds}

Apply \cref{thm:gaussian-product-truncation} to $K_t$ using
\cref{lem:uniform-completion}, and write $c_*$ for the resulting value of
$c_{\mathrm{tr}}$.  Uniformly for $t\in\mathcal I$, the theorem gives
$\|K_t[T]\|_{\mathrm{row}}\le1/8$ whenever $|T|\le c_*n$.
The unused margin between $1/8$ and $1/4$ will accommodate
finite-precision errors.
All constants in this section may depend on the fixed parameters in
\cref{thm:matrix}.

For an integer $R\ge2$ and a parameter $t\in\mathcal I$, define
the truncated approximation
\[
  Z_{E,R}(t)=
  \det(I+tE)^{-1/2}\sum_{|S|<R}W_{K_t}(S).
\]

By \cref{eq:diffuse-gaussian-truncation-bound,lem:uniform-completion}, there are
constants $C,c>0$ and $n_0$ such that
\[
  |Z_E(t)-Z_{E,R}(t)|
  \le \left(C\sqrt{\frac Rn}\right)^R+e^{-cn}
\]
whenever $n\ge n_0$, $t\in\mathcal I$, and $2\le R\le c_*n$.
Here we also used the uniform upper bound on the determinant factor from
\cref{lem:uniform-completion}.

Because the expansion over coordinate subsets is signed, the additive estimate must be
paired with a uniform positive lower bound.
Indeed, \cref{eq:diffuse-gaussian-near-one}, applied uniformly to $K_t$, gives
$\mathcal Z_{K_t}=1+O(n^{-1})$.  Together with the determinant bounds in
\cref{lem:uniform-completion}, this gives constants
$0<a_{\mathrm{pos}}\le b_{\mathrm{pos}}<\infty$ such that
\begin{equation}
  a_{\mathrm{pos}}\le Z_E(t)\le b_{\mathrm{pos}}
  \quad(t\in\mathcal I),
  \qquad
  a_{\mathrm{pos}}\le Q_E\le b_{\mathrm{pos}}
  \label{eq:positive-middle}
\end{equation}
for all sufficiently large $n$, after adjusting the two constants.  To
obtain the bounds for $Q_E$, average the first pair of inequalities over
$U\in[r_0,r_1]$ and use \cref{prop:radial-localization} for the discarded
Gamma tails.

\begin{corollary}[Perfect-matching existence]
\label{cor:spectral-existence}
For every fixed choice of the parameters in
\cref{thm:scaling-criterion}, there is an $n_0$ such that the support of
every input covered by that theorem with $n\ge n_0$ contains a perfect
matching.
\end{corollary}

\begin{proof}
The entropy-scaled matrix is $X=P_0+E$ and has the same support as the
input.  By \cref{eq:positive-middle},
\(\haf(X)=\haf(P_0)Q_E>0\) once $n\ge n_0$.  Thus the support contains a
perfect matching.  In particular, only the exact branch of the general
criterion can encounter a zero hafnian.
\end{proof}

We next collect the quantitative estimates used by the algorithm.  Recall the Gamma
density $p_{n,\tau}$ from \cref{eq:gamma-density}.  For an integer $R\ge2$, define
\[
  \varphi_R(u)=p_{n,\tau}(u)Z_{E,R}(u^{-1}),
  \qquad
  Q_R=\int_{r_0}^{r_1}\varphi_R(u)\,du.
\]
Thus $\varphi_R$ is the truncated Gamma-average integrand, and $Q_R$ is its integral
over the retained interval.

There are constants $C,c>0$, $0<\eta_0\le c_*/2$, and $n_0$ such that,
whenever $n\ge n_0$ and $2\le R\le2\eta_0 n$,
\begin{equation}
  |Q_E-Q_R|
  \le\left(C\sqrt{\frac Rn}\right)^R+e^{-cn}.
  \label{eq:direct-radial-truncation}
\end{equation}
Moreover,
\[
  C^{-1}\le Q_E\le C,
  \qquad
  \|\varphi_R'\|_{L^\infty([r_0,r_1])}
  \le n^C\left(\frac{3n}{R}\right)^{CR}.
\]
At a specified point $u\in[r_0,r_1]$, the expression $\varphi_R(u)$ has a
real-arithmetic evaluation circuit of size at most
$n^C(3n/R)^{CR}$.  The truncation estimate follows from localization and
the uniform estimate above.  For the remaining bounds, the resolvent
formulas show that $K_t$ and $K_t'$ have entries $O(n^{-1})$.  If
$s=|T|<R$, then $\|K_t[T]\|_{\mathrm{row}}\le1/4$ and
$\|(I-K_t[T])^{-1}\|_{\mathrm{row}}\le4/3$,
and the derivatives of the shifted covariance and mean are $O(1)$.
Differentiating the fixed-subset formula and its recurrence therefore
bounds a fixed $T$ contribution and its derivative by
$e^{O(s)}\poly(s)$, with $2^s\poly(s)$ operations.
Inclusion--exclusion costs $3^s\poly(s)$, and summing over $|S|<R$ gives
the claimed bound.  Finally,
$\|p_{n,\tau}\|_\infty=O(\sqrt n)$ and
$\|p_{n,\tau}'\|_\infty=O(n^{3/2})$ by Stirling's formula, so the same
bound holds for $\varphi_R'$.

\subsubsection{Choosing the cutoff and quadrature}

We now choose the subset cutoff so that the analytic tail is smaller than
the requested error, and then choose the quadrature mesh from the derivative
bound.

We may assume that $C\ge12$.  Put
\[
  b=\left\lceil\log_2\frac{C^2}{\eps}\right\rceil,
\]
and choose $\eta>0$ no larger than $\eta_0$, $c/(4\log2)$, and the
constant supplied by \cref{lem:generic-cutoff} for
$\alpha=1/2$, $C_2=C$, and $c_{\mathrm{tr}}=2\eta_0$.  We may decrease
$\eta$ further without changing that lemma.
If $n<n_0$ or $b\ge\eta n$, we use exact evaluation.
Otherwise, let
$R\ge2$ be the least integer satisfying
\begin{equation}
  \left(\frac{C^2R}{n}\right)^R\le2^{-2b}.
  \label{eq:simple-cutoff}
\end{equation}
The algorithm finds this least integer by testing successive values of
$R$.
The cutoff condition is equivalently
$(C\sqrt{R/n})^R\le2^{-b}$.  Hence \cref{lem:generic-cutoff} shows that
$R$ exists, satisfies $R\le2\eta_0n$, and obeys
\begin{equation}
  R\log\frac{en}{R}=O(b+\log n).
  \label{eq:R-runtime}
\end{equation}
The cutoff and the choice of $\eta$ also give
$\bigl(C\sqrt{R/n}\bigr)^R\le2^{-b}$ and $e^{-cn}\le2^{-4b}$,
and the derivative and circuit bounds in
the preceding display are $n^{O(1)}2^{O(b)}$.

For the quadrature, divide $[r_0,r_1]$ into $L$ equal subintervals and use
their midpoints $u_1,\ldots,u_L$.  Choose $L$ to be a sufficiently large
constant multiple of $2^b n^C(3n/R)^{CR}$, rounded up to an integer.
With $h=(r_1-r_0)/L$, put
\begin{equation}
  Q_{R,L}=h\sum_{\ell=1}^L \varphi_R(u_\ell).
  \label{eq:midpoint-rule}
\end{equation}
The elementary Lipschitz midpoint estimate and
the derivative estimate above give
\(|Q_{R,L}-Q_R|\le2^{-b}\).
Thus the truncation, Gamma tail, and quadrature errors together are at
most $3\cdot2^{-b}$.  The number of nodes and the work at each node are
polynomial in $n$ and $1/\eps$ by \cref{eq:R-runtime}.

\subsubsection{The exact branch}

For a symmetric weighted matrix $Y$ and a vertex set $S$, let $F_Y(S)$ be
the hafnian of the principal matrix $Y[S]$, with $F_Y(\varnothing)=1$.
For nonempty even $S$, choose its least vertex $v$ and use the standard
subset recurrence
\[
  F_Y(S)=\sum_{u\in S\setminus\{v\}}Y_{uv}
  F_Y(S\setminus\{u,v\}).
\]
The recurrence has $2^{O(n)}$ states.  If $b\ge\eta n$, then
$2^{O(n)}\le2^{O_{\beta,\kappa}(b)}=\eps^{-O_{\beta,\kappa}(1)}$.  The other
exact branch has $n<n_0$, a fixed threshold depending only on $\beta$ and
$\kappa$.  Thus both branches have polynomial real-arithmetic complexity
in $n$ and $1/\eps$.

For rational inputs the recurrence is applied to the original matrix rather
than to the generally irrational entropy scaling.  Its bit implementation,
and the zero-output case, are recorded in \cref{app:bit-complexity}.

\begin{proof}[Proof of \cref{thm:matrix}]
We verify accuracy and arithmetic cost in the approximation branch, then
invoke the subset recurrence for the complementary parameter range.  In
the approximation branch, evaluate
\cref{eq:midpoint-rule}.  At each
node, compute the weights $W_{K_t}(S)$ for $|S|<R$ by
\cref{prop:fixed-support}.  The number of arithmetic operations is bounded by
\[
  L\sum_{s<R}\binom ns3^s\poly(n)
  \le
  L\exp\left(O\left(R\log\frac{en}{R}\right)\right)\poly(n).
\]
By \cref{eq:R-runtime} and the choice of $L$, this is
$n^{O(1)}2^{O(b)}=n^{O(1)}\eps^{-O(1)}$.  The estimates of the preceding
subsection show that the resulting value differs from $Q_E$ by at most
$3\cdot2^{-b}\le3\eps/C^2$.  Since $Q_E\ge C^{-1}$, this is a relative
error of at most $3\eps/C\le\eps/4$, which lies within the requested
factor $e^{\pm\eps}$.
Multiplication by the explicit value of $\haf(P)$ gives the requested
hafnian.  The exact branch has already been discussed, so the theorem
follows.
\end{proof}

The preceding proof is in real arithmetic.  The entropy scaling and the
truncated formula are stable under polynomially small perturbations, and the
matrices inverted by the evaluators remain uniformly well conditioned.
Consequently, polynomially many working bits suffice.  The complete
bit-model argument is given in \cref{app:bit-complexity}.

We finish the matching section by deriving the main theorem and the three
corollaries stated in the introduction.

\begin{proof}[Proof of \cref{thm:main}{\rm (a)}]
All bounded exceptional dimensions arising in the reduction are handled
by exact evaluation.  Above the resulting fixed threshold,
\cref{lem:edge-extendability} shows that every edge of the support graph
belongs to a perfect matching.
The factor bounds in \cref{eq:graph-factor-bounds} hold with
$c_{\mathrm{lo}}=\theta\sqrt{2\gamma}$ and
$c_{\mathrm{hi}}=(\theta\sqrt{2\gamma})^{-1}$.
Moreover, \cref{lem:centered-gap} gives
$\lambda_2(L_G)\ge2\gamma n$ and $Q_G\succeq2\gamma nI$.
Thus \cref{thm:scaling-criterion}, with $q'=q=2\gamma$, gives the stated
algorithm and bit complexity.
\end{proof}

\begin{proof}[Proof of \cref{thm:main}{\rm (b)}]
The finitely many dimensions below the threshold in
\cref{prop:bipartite-scaling} are handled exactly.  In the remaining
dimensions, that proposition gives total support, the factor bounds
\cref{eq:bipartite-factor-box}, and
$\sigma_2(Y^*)\le1-4\gamma^3\theta^6$.  Therefore
\cref{thm:bipartite-scaling-criterion}, with the fixed factor bounds from
\cref{eq:bipartite-factor-box} and
$\kappa_{\mathrm b}=4\gamma^3\theta^6$, gives the stated algorithm and bit
complexity.
\end{proof}

\begin{proof}[Proof of \cref{cor:positive-dense}]
Apply \cref{thm:main} with, for example, $\gamma=1/4$.  Complete supports
satisfy the required degree inequalities in all sufficiently large
dimensions, and the remaining dimensions are handled exactly.
\end{proof}

\begin{proof}[Proof of \cref{cor:near-perfect}]
Let $n=|V(G)|$ and take $\gamma'=\gamma/2$.  For every vertex $v$, the graph
$G-v$ satisfies
$\delta(G-v)\ge(1/2+\gamma)n-1\ge(1/2+\gamma')(n-1)$
when $n$ is sufficiently large in terms of $\gamma$.  Apply
\cref{thm:main}{\rm(a)} to every $G-v$ and add the estimates.  Each
near-perfect matching occurs once, according to its unmatched vertex.  A sum
of positive $e^{\pm\eps}$-approximations has the same relative guarantee.
The remaining bounded values of $n$ are handled exactly.
\end{proof}

\begin{proof}[Proof of \cref{cor:constant-scale-entropy}]
The entropy normalization and the definition of $Q_E$ give
\[
  \haf(A)=e^{h_A(G)}\haf(P_0)Q_E,
  \qquad
  \haf(P_0)=(n-1)!!(n-1)^{-n/2}.
\]
\Cref{eq:positive-middle} bounds $Q_E$ above and below by
positive constants that depend only on $\gamma$ and $\theta$.  Decreasing
the lower constant $a_{\gamma,\theta}$ and increasing the upper constant
$b_{\gamma,\theta}$ handles the finitely many values of $n$ sent to the
exact branch.  Every admissible support has a perfect matching by Dirac's
theorem.  There are only finitely many such supports in these dimensions,
and compactness of $[\theta,1]^{E(G)}$ gives uniformity over the weights for
each support.  This proves the first assertion.  Stirling's formula gives
$\log\haf(P_0)=-n/2+O(1)$.
Taking logarithms proves the equivalent form.
\end{proof}

\section{The zero-field Ising model with diffuse couplings}
\label{sec:ising}

The Ising application is shorter because no entropy scaling or Gamma average
is needed.  Hubbard--Stratonovich converts the spin sum into a product of
hyperbolic cosines.  Quadratic completion then produces the same
Gaussian-product form as in the matching argument, but with growth in the
imaginary rather than the real direction.  We first verify the required
one-variable bound.  We then prove the representation, bound the completed
moment matrix, and evaluate the retained subsets.

Let $J$ be a real symmetric $n$ by $n$ matrix with zero diagonal.  We use the
normalization
\begin{equation}
  Z(J)
  =\sum_{\sigma\in\{\pm1\}^n}
    \exp\!\left(\frac12\sigma^{\mathsf T}J\sigma\right).
  \label{eq:ising-partition}
\end{equation}
The diagonal restriction is harmless.  An arbitrary diagonal contributes
the known factor $\exp(\operatorname{tr}J/2)$.

The one-variable functions needed below are
$g_{\mathrm{Is}}(z)=e^{-z^2/2}\cosh z$ and
$f_{\mathrm{Is}}(z)=g_{\mathrm{Is}}(z)-1$.

\begin{lemma}[Ising scalar factor]
\label{lem:ising-factor}
The function $g_{\mathrm{Is}}$ satisfies
\begin{align}
  |g_{\mathrm{Is}}(x+iy)|&\le e^{y^2/2},
  \label{eq:ising-envelope}\\
  f_{\mathrm{Is}}(z)&=-\frac{z^4}{12}+O(z^6)
  \qquad(z\longrightarrow0).
  \label{eq:ising-quartic-zero}
\end{align}
It is therefore admissible with envelope weights
$a_{\mathrm R}=0$ and $a_{\mathrm I}=1/2$, vanishing order $\ell=4$,
and local constants $C_f=1/8$ and $\delta_f=1/2$.
\end{lemma}

The elementary proof is given in
Appendix~\ref{subsec:ising-factor-proof}.

\subsection{Gaussian completion}

We now remove the spin sum and simultaneously resum the quadratic part of
the resulting Gaussian product.

\begin{lemma}[Hubbard--Stratonovich representation and completion]
\label{lem:ising-completion}
Suppose that $\lambda_{\max}(J)<1$ and define
\begin{equation}
  K=J(I-J)^{-1}.
  \label{eq:ising-kernel}
\end{equation}
Then
\begin{equation}
  Z(J)
  =2^n\det(I-J)^{-1/2}
    \mathbb E_K\prod_{i=1}^n g_{\mathrm{Is}}(\phi_i).
  \label{eq:ising-completed}
\end{equation}
The square root is the positive one, and the expectation in
\cref{eq:ising-completed} converges absolutely.
\end{lemma}

\begin{proof}
If $x$ is a standard real Gaussian vector, then, for every complex vector
$v$, one has $\mathbb E e^{v^{\mathsf T}x}=e^{v^{\mathsf T}v/2}$.  Take
$v=A_J\sigma$.  Since $A_J$ is symmetric and $A_J^2=J$, we have
$v^{\mathsf T}v=\sigma^{\mathsf T}J\sigma$.  Summing the resulting identity
over the spin configurations gives
$Z(J)=2^n\mathbb E_x\prod_{i=1}^n\cosh((A_Jx)_i)$.
If $\psi=A_Jx$, then
$\sum_i\psi_i^2=x^{\mathsf T}Jx$.  The definition of
$g_{\mathrm{Is}}$ therefore gives
\[
  \mathbb E_x\prod_i\cosh(\psi_i)
  =\mathbb E_x e^{x^{\mathsf T}Jx/2}
     \prod_i g_{\mathrm{Is}}(\psi_i).
\]
The assumption on $J$ makes $I-J$ positive definite.  The change of
variables $y=(I-J)^{1/2}x$ absorbs the quadratic exponential into the
Gaussian density and contributes $\det(I-J)^{-1/2}$.  Spectral calculus
gives $A_J(I-J)^{-1/2}=A_K$, where $K=J(I-J)^{-1}$,
which proves the identity.

It remains only to justify absolute convergence.  Before the change of
variables,
$\prod_i|\cosh((A_Jx)_i)|\le\exp(x^{\mathsf T}J_+x/2)$,
and $\lambda_{\max}(J_+)<1$.  After the change of variables,
\cref{eq:ising-envelope} gives the majorant
$\exp(y^{\mathsf T}K_-y/2)$.  If $-a<0$ is an eigenvalue of $J$, the
corresponding eigenvalue of $K_-$ is $a/(1+a)<1$.  The second Gaussian
integral is therefore absolutely convergent as well.
\end{proof}

\subsection{The transformed moment matrix}

The completed moment matrix retains the diffuse entrywise scale and also
satisfies the spectral margin needed below.

\begin{lemma}[Bounds for the completed moment matrix]
\label{lem:ising-kernel}
Fix $\beta>0$ and $0<\kappa\le1$.  Suppose that $J$ is real symmetric,
has zero diagonal, and satisfies
\begin{equation}
  \max_{i,j}|J_{ij}|\le\frac{\beta}{n},
  \qquad
  \lambda_{\max}(J)\le1-\kappa.
  \label{eq:ising-hypotheses}
\end{equation}
For the matrix $K$ in \cref{eq:ising-kernel}, put
$\beta_K=\beta+\beta^2/\kappa$ and $d_K=(1+\beta)^{-1}$.
Then
\begin{equation}
  \max_{i,j}|K_{ij}|\le\frac{\beta_K}{n},
  \qquad I-K_-\succeq d_KI.
  \label{eq:ising-kernel-bounds}
\end{equation}
The negative eigenvalues of $K$ lie in $[-\beta/(1+\beta),0]$.  Finally,
\begin{equation}
  |\log\det(I-J)|\le\frac{\beta^2}{2\kappa}.
  \label{eq:ising-determinant-bound}
\end{equation}
\end{lemma}

\begin{proof}
The entrywise hypothesis gives $\|J\|\le\|J\|_{\mathrm{row}}\le\beta$.
Using $K=J+J(I-J)^{-1}J$
and $\|(I-J)^{-1}\|\le\kappa^{-1}$, we obtain
\[
  |(J(I-J)^{-1}J)_{ij}|
  \le \|Je_i\|_2\,\|(I-J)^{-1}\|\,\|Je_j\|_2
  \le\frac{\beta^2}{\kappa n}.
\]
Here $e_i$ and $e_j$ are standard coordinate vectors.
This proves the entrywise part of \cref{eq:ising-kernel-bounds}.

The map from the spectrum of $J$ to that of $K$ is
$\lambda\mapsto\lambda/(1-\lambda)$.  If $\lambda=-a\le0$, then the corresponding
eigenvalue of $K$ is $-a/(1+a)$.  Since $a\le\beta$, this proves the
remaining spectral assertions and $I-K_-\succeq d_KI$.  In view of
\cref{eq:ising-envelope}, these are precisely the entrywise and Gaussian
integrability margins required by \cref{thm:gaussian-product-truncation}.

For the determinant estimate, let $\lambda_1,\ldots,\lambda_n$ be the
eigenvalues of $J$.  The zero diagonal gives $\sum_i\lambda_i=0$, and hence
$\log\det(I-J)=\sum_i(\log(1-\lambda_i)+\lambda_i)$.
For $0\le\lambda\le1-\kappa$, integration of
$t/(1-t)$ gives
$0\le-\log(1-\lambda)-\lambda\le\lambda^2/(2\kappa)$.
For $\lambda=-a\le0$, the inequality
$0\le a-\log(1+a)\le a^2/2$ gives the same bound.  Thus
\[
  |\log\det(I-J)|
  \le\frac{\|J\|_{\mathrm F}^2}{2\kappa}
  \le\frac{\beta^2}{2\kappa},
\]
as claimed.
\end{proof}

Combining \cref{eq:ising-envelope,eq:ising-quartic-zero} with
\cref{lem:ising-kernel} shows that the pair $(g_{\mathrm{Is}},K)$ satisfies
the hypotheses of \cref{thm:gaussian-product-truncation}, with constants
depending only on $\beta$ and $\kappa$.  Indeed, its envelope matrix is
$K_{\mathrm{env}}=K_-/2$, so the required spectral margin is precisely
$I-K_-\succeq d_KI$.  In particular, writing
$\mathcal Z_K=\mathbb E_K\prod_{i=1}^n g_{\mathrm{Is}}(\phi_i)$, we have
\begin{equation}
  \mathcal Z_K=1+O_{\beta,\kappa}(1/n).
  \label{eq:ising-normalized-asymptotic}
\end{equation}

\subsection{Evaluation of retained subsets}

For a subset $T\subseteq[n]$, define
$G_K(T)=\mathbb E_{K[T]}\prod_{i\in T}g_{\mathrm{Is}}(\phi_i)$.
The retained contribution has an explicit sum over the spin configurations
on $T$.

\begin{lemma}[Retained-subset evaluation]
\label{lem:ising-retained-term}
Let $K_T$ be a real symmetric $t$ by $t$ matrix with
$\|K_T\|_{\mathrm{row}}\le1/4$.  Then
\begin{equation}
  \mathbb E_{K_T}\prod_{i=1}^t g_{\mathrm{Is}}(\phi_i)
  =\det(I+K_T)^{-1/2}2^{-t}
    \sum_{\sigma\in\{\pm1\}^t}
    \exp\!\left(
      \frac12\sigma^{\mathsf T}K_T(I+K_T)^{-1}\sigma
    \right).
  \label{eq:ising-subset-evaluator}
\end{equation}
In particular, the right-hand side is a positive real number, so the
full-product terms $G_K(T)$ used below are positive.  It is computable using
$2^t\operatorname{poly}(t)$ arithmetic operations.  Every exponent in
\cref{eq:ising-subset-evaluator} has absolute value at most $t/6$.
\end{lemma}

\begin{proof}
Let $\phi=A_{K_T}x$.  Expanding the hyperbolic cosines gives
\[
  \prod_i g_{\mathrm{Is}}(\phi_i)
  =2^{-t}\sum_{\sigma\in\{\pm1\}^t}
    \exp\!\left(-\frac12x^{\mathsf T}K_Tx
                    +\sigma^{\mathsf T}A_{K_T}x\right).
\]
Since $I+K_T$ is positive definite, Gaussian integration with a complex
linear term gives \cref{eq:ising-subset-evaluator}.  Here we used
$A_{K_T}(I+K_T)^{-1}A_{K_T}=K_T(I+K_T)^{-1}$, which follows from
spectral calculus.
All matrices in the expression on the right of
\cref{eq:ising-subset-evaluator} are real, so the value is positive and
real.

The Neumann series gives
$\|(I+K_T)^{-1}\|_{\mathrm{row}}\le4/3$, and hence
$\|K_T(I+K_T)^{-1}\|_{\mathrm{row}}\le1/3$.  Therefore
\[
  \left|\frac12\sigma^{\mathsf T}K_T(I+K_T)^{-1}\sigma\right|
  \le\frac t2\|K_T(I+K_T)^{-1}\|_{\mathrm{row}}
  \le\frac t6.
\]
\end{proof}

The subset coefficients in the truncation theorem are
$W_K(S)=\mathbb E_{K[S]}\prod_{i\in S}f_{\mathrm{Is}}(\phi_i)$.
By inclusion--exclusion and \cref{lem:gaussian-marginalization},
\begin{equation}
  W_K(S)=\sum_{T\subseteq S}(-1)^{|S|-|T|}G_K(T).
  \label{eq:ising-inclusion-exclusion}
\end{equation}
Thus a direct evaluation of $W_K(S)$ takes
$3^{|S|}\operatorname{poly}(|S|)$ operations: after a set $T$ is chosen,
\cref{eq:ising-subset-evaluator} sums over its $2^{|T|}$ spin
configurations, and $\sum_{T\subseteq S}2^{|T|}=3^{|S|}$.

\subsection{The approximation scheme}

We can now prove the main result.

\begin{proof}[Proof of \cref{thm:ising}]
The asymptotic statements follow at once from
\cref{eq:ising-completed,eq:ising-determinant-bound,%
eq:ising-normalized-asymptotic}.  We prove the algorithmic assertion.

Apply \cref{thm:gaussian-product-truncation} with the fixed Ising factor, the
entrywise constant $\beta_K$, and the margin $d_K$.  Denote its constants
by $C_1,c_{\mathrm{tr}},a_{\mathrm{tail}}>0$ and $n_0$.  Increase $n_0$, if
necessary, so that
\begin{equation}
  \mathcal Z_K\ge\frac12
  \qquad(n\ge n_0).
  \label{eq:ising-positive-normalization}
\end{equation}
Set $C_{\mathrm{cut}}=2C_1$ and
$b=\lceil\log_2(16/\eps)\rceil$.
Fix a sufficiently small constant $\eta>0$, depending only on
$\beta$ and $\kappa$, for which \cref{lem:generic-cutoff} applies with
$\alpha=1$ and $C_2=C_{\mathrm{cut}}$, and such that
$\eta\le a_{\mathrm{tail}}/(3\log2)$.

We first dispose of the enumeration branch.  If $n<n_0$ or
$b\ge\eta n$, evaluate the positive sum in
\cref{eq:ising-partition} directly.  In the second case,
$2^n\le2^{b/\eta}$, while in the first case $n$ is bounded by a constant.
Moreover,
$|\sigma^{\mathsf T}J\sigma|/2\le\beta n/2$.  Standard evaluation of the
positive exponential terms to a common relative error below $\eps/4$
therefore gives the required estimate in
$2^{O_{\beta,\kappa}(b)}\operatorname{poly}(L_J+b)$ bit operations.

Suppose from now on that $n\ge n_0$ and $b<\eta n$.  Let $R\ge2$ be the
least integer satisfying
\begin{equation}
  \left(C_{\mathrm{cut}}\frac Rn\right)^R\le2^{-b}.
  \label{eq:ising-cutoff}
\end{equation}
By \cref{lem:generic-cutoff}, after decreasing $\eta$ once, this integer
exists and satisfies $R\le c_{\mathrm{tr}}n$.  Minimality gives
\begin{equation}
  R\log\frac{en}{R}=O_{\beta,\kappa}(b+\log n).
  \label{eq:ising-cutoff-size}
\end{equation}
Since $R\ge2$, replacing $C_1$ by $C_{\mathrm{cut}}=2C_1$ divides the first
term of the truncation bound by at least four.  Our choice of $\eta$ gives
$e^{-a_{\mathrm{tail}}n}\le2^{-3b}$.  Hence
\begin{align*}
  \left|
    \mathcal Z_K-\sum_{|S|<R}W_K(S)
  \right|
  &\le
    \left(C_1\frac Rn\right)^R+e^{-a_{\mathrm{tail}}n} \\
  &\le2^{-b-2}+2^{-3b}
  \le2^{-b-1}.
\end{align*}
Compute the retained sum by
\cref{eq:ising-subset-evaluator,eq:ising-inclusion-exclusion} to absolute
error at most $2^{-b-2}$, and call the result
$\widetilde{\mathcal Z}_K$.  Then
\begin{equation}
  |\widetilde{\mathcal Z}_K-\mathcal Z_K|
  \le3\cdot2^{-b-2}.
  \label{eq:ising-normalized-error}
\end{equation}
Approximate the positive prefactor
$2^n\det(I-J)^{-1/2}$ to relative error at most $2^{-b-2}$ and multiply it
by $\widetilde{\mathcal Z}_K$.  From
\cref{eq:ising-positive-normalization,eq:ising-normalized-error}, the
relative error in $\widetilde{\mathcal Z}_K$ is at most $3\cdot2^{-b-1}$.
In particular, $\widetilde{\mathcal Z}_K>0$, because
$\mathcal Z_K\ge1/2$ and $3\cdot2^{-b-2}<1/2$.
The relative error in the final product is therefore less than
$2^{1-b}\le\eps/8$.  Since $0<\eps<1$, this is contained in the interval
$[e^{-\eps},e^\eps]$.

Finally, \cref{eq:ising-cutoff-size} gives
$\sum_{s<R}\binom ns3^s=n^{O_{\beta,\kappa}(1)}2^{O_{\beta,\kappa}(b)}$.
Together with \cref{lem:ising-retained-term}, this proves the polynomial
real-arithmetic bound.  The bit implementation is given in
\cref{app:bit-complexity}, completing the proof.
\end{proof}

\section{Scope and open problems}
\label{sec:scope}

We close with three limitations of the present method and the questions
they suggest.

\overviewstage{1}{The truncation principle}
The entrywise $O(1/n)$ hypothesis is intrinsic to the present truncation
argument and is not implied by the spectral margin, as simple rank-one
examples show.  This prevents a direct application to sparse expanders,
where a graph-specific resummation or correlation estimate would be needed.
Both applications also eliminate a linear mode, through $E\1=0$ for
matchings and zero external field for Ising.  We do not know whether a
surviving finite-rank mode can be integrated separately while the theorem is
applied to the centered fluctuations.  Such an extension would require
uniform envelope and evaluation bounds that we do not yet have.

\overviewstage{2}{Perfect matchings}
The present method does not reach the exact graph Dirac threshold.  There
the entropy scaling may place order-one weight on an edge.  Indeed, split
$n=2m$ vertices into two $m$-vertex sets $U$ and $V$, include all cross
edges, all edges within $V$, and one edge within $U$.  If $t$ is the
entropy-maximizing weight of that last edge, symmetry gives
$t/(1-t)=\sqrt{(m-1)/(2m)}$, so $t\to\sqrt2-1$.  Separately, a copy of
$K_{m,m}$ with one added edge in each part shows that the smallest scaled
eigenvalue can approach $-1$.  Whether a deterministic FPTAS exists at the
endpoint for all Dirac graphs remains open.

A separate argument gives a deterministic FPTAS for the structured endpoint
family obtained by replacing the edge inside $U$ by a graph $F$ of maximum
degree at most a fixed constant $\Delta$.  Let $m_k(F)$ denote the number of
$k$-edge matchings in $F$, and let $\mu$ be the maximum size of a matching
in $F$.  Then
$\#\PM=m!\sum_{k=0}^{\mu} m_k(F)/(2^k k!)$.  If
$\mu\ge\max\{8e^2\Delta,2\log_2(8/\eps)\}$, the bound
$|E(F)|\le 2\Delta\mu$ makes the terms with $k>\mu/2$ contribute at most
$\eps/4$ of the sum, while the algorithm of Jain, Perkins, Sah, and Sawhney
\cite{JainPerkinsSahSawhney2022}, with $\delta=1/2$, approximates all terms
with $k\le\mu/2$.  Otherwise, a dynamic program over the endpoints of a
maximum matching evaluates the sum exactly in
$2^{O(\mu)}\poly(m)=\poly(m,1/\eps)$ time.

Below the half-degree threshold,
\cref{cor:pseudorandom-supports} controls unweighted pseudorandom supports,
but the analogous weighted statement remains open in the full
degree--codegree range.  Our counting algorithm also does not
immediately give a sampler, since conditioning on $k$ matching edges reduces
the guaranteed degree surplus from $\gamma n$ to $\gamma n-k$.

\overviewstage{3}{The Ising model}
The bounded-remainder conclusion in \cref{thm:ising} also requires
diffuseness.  A single fixed bond of strength $q\in(0,1)$ gives the
normalized factor $\sqrt{1-q^2}\cosh q$, independently of the number of
isolated spins, so the $1+O(1/n)$ conclusion fails without the entrywise
hypothesis.  Deterministic approximation under a spectral hypothesis alone
remains open.  Our argument also does not reach the $n^{-1/2}$ interaction
scale of dense spin glasses: the effective entrywise parameter grows like
$\sqrt n$, and the small-support bound no longer survives the sum over
supports.

\bibliographystyle{alpha}
\bibliography{diffuse_gaussian_counting}

\clearpage
\appendix

\section{Proofs of the scaling-criterion consequences}
\label[appendix]{app:scaling-consequences}

This appendix supplies the deductions deferred from the consequences of
the scaling criteria.  The first proof treats the regular cases, whose
stochastic normalization is explicit.  The remaining argument establishes
the symmetric degree--codegree consequence.

\subsection{Regular inputs}

\begin{proof}[Proof of \cref{cor:dense-regular-expanders}]
For part~\textup{(a)}, the maximum-entropy scaling is $X=A_G/d$, with
$r_v=d^{-1/2}$ and
$n^{-1/2}\le r_v\le c^{-1/2}n^{-1/2}$.  The spectral hypothesis gives
$\lambda_2(L_G)\ge\kappa cn$ and $Q_G\succeq\kappa cnI$.
Apply \cref{thm:scaling-criterion} with
$\theta=1$, $c_{\mathrm{lo}}=1$, $c_{\mathrm{hi}}=c^{-1/2}$, and $q'=q=\kappa c$.

For part~\textup{(b)}, regularity gives total support and the Sinkhorn
scaling $Y^*=B/d$, with $r_i=c_j=d^{-1/2}$.  The bounds are again between
$m^{-1/2}$ and $c^{-1/2}m^{-1/2}$.  Apply
\cref{thm:bipartite-scaling-criterion} with
$\theta=1$, $c_{\mathrm{lo}}=1$, $c_{\mathrm{hi}}=c^{-1/2}$, and
$\kappa_{\mathrm b}=\kappa$.
\end{proof}

\subsection{Symmetric degree and codegree bounds}

\begin{proof}[Proof of \cref{cor:pseudorandom-supports}]
The edge-extension hypothesis makes the maximum-entropy scaling positive
on every edge.  Write it as $X_{uv}=r_ur_v$, and put
\(d_-=(1-\rho)p, \qquad d_+=(1+\rho)p, \qquad R_0=\frac{d_-}{\eta-2\rho p}\).
At vertices $a,b$ carrying the largest and smallest factors, subtraction
of the row equations and the degree bounds give
$1/r_b-1/r_a\le(d_+-\eta)nr_a$ and
$1/(d_+n)\le r_ar_b\le1/(d_-n)$.
Multiplying the first inequality by $r_b$ and using the second yields
$r_{\max}/r_{\min}\le R_0$.  Consequently,
$1/\sqrt{R_0d_+n}\le r_v\le\sqrt{R_0/(d_-n)}$ for every vertex $v$,
and hence $1/(R_0d_+n)\le X_{uv}\le R_0/(d_-n)$ on every edge.

For the Laplacian gap, sum
$(z_u-z_w)^2\le2(z_u-z_v)^2+2(z_v-z_w)^2$ over common neighbors $v$ of
each pair $u,w$.  Every squared edge difference is counted at most $2n$
times.  Since $\sum_{u<w}(z_u-z_w)^2=n\|z\|_2^2$ for $z\perp\1$, this gives
\[
  z^{\mathsf T}L_Gz\ge\frac{\eta n}{4}\|z\|_2^2,
  \qquad
  \lambda_2(L_G)\ge\frac{\eta n}{4}.
\]
The displayed factor bounds, this Laplacian estimate, and
$Q_G\succeq qnI$ verify the hypotheses of
\cref{thm:scaling-criterion}, which proves the algorithmic claim.

For $G(n,p)$, take $\rho=p/8$, $\eta=p^2/2$, and any fixed $q<p/2$.
Chernoff bounds and the F\"uredi--Koml\'os estimate give the degree,
codegree, and signless-Laplacian conditions with high probability
\cite{FurediKomlos1981}.  For a fixed pair $u,v$, put $s=(n-2)/2$ and
partition the remaining vertices into two $s$-sets.  Hall's condition gives
\begin{equation}
  \Pr\{G(s,s,p)\text{ has no perfect matching}\}
  \le
  \sum_{k=1}^s
  \binom sk\binom{s}{k-1}(1-p)^{k(s-k+1)}
  =e^{-\Omega_p(s)}.
  \label{eq:random-bipartite-hall}
\end{equation}
Pairing $k$ with $s-k+1$ bounds the sum by a geometric series with ratio
$s^2(1-p)^{s/2}$, proving the last equality.  A union bound over all pairs
then shows that every present edge extends to a perfect matching.
\end{proof}

\section{Bipartite scaling details}
\label[appendix]{app:bipartite-scaling}

The permanent reduction uses only standard scaling and spectral estimates.
For reference, we give the details here in the order in which they are
invoked in the main text.

\subsection{Entropy normalization}

\begin{proof}[Proof of \cref{prop:bipartite-normalization}]
Total support lets us average supported permutation matrices to obtain a
feasible point positive on every supported entry.  The boundary argument
from \cref{prop:entropy-normalization} puts $Y^*$ in the same relative
interior, and the row and column multiplier equations give
$Y^*_{ij}=b_{ij}r_ic_j$.  Moreover,
$\sum_i\log r_i+\sum_j\log c_j
=\sum_{i,j}Y^*_{ij}\log(Y^*_{ij}/b_{ij})=-h_B$.
Every supported permutation uses each row and column factor once.
Its term in $\operatorname{per}(B)$ is therefore $e^{h_B}$ times its term
in $\operatorname{per}(Y^*)$.  Summing over supported permutations proves
\cref{eq:bipartite-entropy-normalization}.
\end{proof}

\subsection{Fixed-margin Sinkhorn estimates}

\begin{proof}[Proof of \cref{prop:bipartite-scaling}]
Put $\alpha=1/2+\gamma$, and let $G_B$ denote the bipartite support graph
of $B$.  If $b_{ij}>0$, deleting row $i$ and column $j$
from the support leaves a balanced bipartite graph of minimum degree at least
$\alpha m-1\ge(m-1)/2$ by the assumed lower bound on $m$.  Hall's theorem
therefore extends $ij$ to a perfect matching.  Thus $B$ has total support.

For row sets $S$ and column sets $T$, write
$Y^*(S,T)=\sum_{i\in S,j\in T}Y^*_{ij}$,
$r(S)=\sum_{i\in S}r_i$, and $c(T)=\sum_{j\in T}c_j$.
Let $U_{\rm r}$ and $U_{\rm c}$ consist of the
$\lfloor(1-\alpha)m\rfloor$ largest row and column factors, and let
$W_{\rm r}$ and $W_{\rm c}$ be their complements.  Stochasticity and
$Y^*_{ij}\le r_ic_j$ give
\[
  Y^*(W_{\rm r},W_{\rm c})
  =|W_{\rm r}|-Y^*(W_{\rm r},U_{\rm c})
  \ge |W_{\rm r}|-|U_{\rm c}|
  \ge2\gamma m,
  \qquad r(W_{\rm r})c(W_{\rm c})\ge2\gamma m.
\]

Write $r_{\max}=\max_i r_i$, $r_{\min}=\min_i r_i$, and define
$c_{\max},c_{\min}$ analogously.  The nonneighbors of any row have at most
$|U_{\rm c}|$ elements.  The row and column equations, together with the
preceding mass bound, give the extremal products as follows.  The
neighboring column-factor sum is at least $c(W_{\rm c})$, and symmetrically
the neighboring row-factor sum is at least $r(W_{\rm r})$.  Thus
$r_{\max}\le1/(\theta c(W_{\rm c}))$ and
$c_{\max}\le1/(\theta r(W_{\rm r}))$.  Conversely, $b_{ij}\le1$ gives
$r_i\ge1/(mc_{\max})$ and $c_j\ge1/(mr_{\max})$.  Hence
$r_{\max}c_{\max}\le1/(2\gamma\theta^2m)$ and
$r_{\min}c_{\min}\ge2\gamma\theta^2/m$.
Since $Y^*_{ij}=b_{ij}r_ic_j$ with $b_{ij}\in[\theta,1]$ on the support,
these estimates prove \cref{eq:bipartite-diffuse}.  They also give the
looser uniform range
$2\gamma\theta^3/m\le r_ic_j\le1/(2\gamma\theta^3m)$
on every supported entry.  Any two rows share a supported column, and any
two columns share a supported row.  Taking ratios through these common
entries shows that the ratio of any two row factors, and likewise of any
two column factors, is bounded by a constant depending only on $\gamma$
and $\theta$.

Put $\ell_0=\log(2\gamma\theta^3)$ and
$u_0=\log(1/(2\gamma\theta^3))$.  On every supported entry,
$\ell_0\le x_i^*+y_j^*\le u_0$.  The gauge gives equal row and column
means, while the ratio bound shows that every coordinate differs from its
mean by a constant depending only on $\gamma$ and $\theta$.  Applying
these facts at one supported entry places every coordinate in a fixed
interval depending only on these two parameters.  Exponentiating proves
\cref{eq:bipartite-factor-box}.

The symmetric stochastic matrix $W=Y^*(Y^*)^{\mathsf T}$ controls the
remaining singular values.  Any two distinct rows of $B$ share at least
$2\gamma m$ common support columns.  Thus \cref{eq:bipartite-diffuse} gives
$W_{ii'}\ge8\gamma^3\theta^6/m$ for $i\ne i'$.  For $z\perp\1$,
\[
  z^{\mathsf T}(I-W)z
  =\sum_{i<i'}W_{ii'}(z_i-z_{i'})^2
  \ge8\gamma^3\theta^6\|z\|_2^2.
\]
Hence $\sigma_2(Y^*)^2\le1-8\gamma^3\theta^6$, which implies
\cref{eq:bipartite-gap}.
\end{proof}

\subsection{The bipartite scaling criterion}

\begin{proof}[Proof of \cref{thm:bipartite-scaling-criterion}]
By \cref{prop:bipartite-normalization},
$\operatorname{per}(B)=e^{h_B}\operatorname{per}(Y^*)$.  The symmetric
dilation $E_{\mathrm b}$ constructed in \cref{sec:matching-normalization} has
entries of modulus at most $2\max\{c_{\mathrm{hi}}^2,1\}/n$ and norm at
most $1-\kappa_{\mathrm b}$.  Thus \cref{thm:matrix} applies with
$\beta=2\max\{c_{\mathrm{hi}}^2,1\}$ and $\kappa=\kappa_{\mathrm b}$.
The analysis in \cref{app:bit-complexity} justifies evaluating the scaling and the
truncated formula to polynomial accuracy, after which we restore
$e^{h_B}$.  Bounded exceptional dimensions are evaluated exactly.
\end{proof}

\subsection{Degree and codegree bounds for bipartite supports}

\begin{proof}[Proof of \cref{cor:pseudorandom-bipartite}]
Let $Y^*_{ij}=b_{ij}r_ic_j$ be the gauge-fixed Sinkhorn scaling, and write
$R_{\rm r}=r_{\max}/r_{\min}$ and $R_{\rm c}=c_{\max}/c_{\min}$.
Comparing the equations at rows with extremal factors, using their common
neighborhood and then interchanging rows and columns, gives
\[
  R_{\rm r}\le\theta^{-1}+\lambda_\theta R_{\rm c},
  \qquad
  R_{\rm c}\le\theta^{-1}+\lambda_\theta R_{\rm r},
  \qquad
  \lambda_\theta=\frac{d_+-\eta}{\theta d_-}.
\]
Here the first inequality follows from
$1/r_{\min}\le1/(\theta r_{\max})+(d_+-\eta)mc_{\max}$ and
$r_{\max}c_{\min}\le(\theta d_-m)^{-1}$.  Put
$R=\max\{R_{\rm r},R_{\rm c}\}$.
The assumption \cref{eq:bipartite-pseudorandom-condition} says that
$\lambda_\theta<1$.  The two preceding inequalities give
$R\le1/\theta+\lambda_\theta R$, and hence
$R\le1/(\theta(1-\lambda_\theta))=d_-/(\eta-d_++\theta d_-)$.
Thus
\begin{equation}
  \max\{R_{\rm r},R_{\rm c}\}
  \le R_\theta,
  \qquad
  R_\theta:=\frac{d_-}{\eta-d_++\theta d_-}.
  \label{eq:bipartite-pseudorandom-ratio}
\end{equation}

The row equations at the minimum and maximum row factors give
$1/(d_+m)\le r_{\min}c_{\max}$ and
$r_{\max}c_{\min}\le1/(\theta d_-m)$.
Using \cref{eq:bipartite-pseudorandom-ratio}, every supported entry of
$Y^*$ therefore satisfies
\begin{equation}
  \frac{\theta}{R_\theta d_+m}
  \le Y^*_{ij}\le
  \frac{R_\theta}{\theta d_-m}.
  \label{eq:bipartite-pseudorandom-diffuse}
\end{equation}

The gauge gives a common geometric mean $s_{\rm geom}$ for the row and
column factors.  The ratio and extremal-product bounds imply
$1/(R_\theta d_+m)\le s_{\rm geom}^2\le
R_\theta/(\theta d_-m)$ and therefore
$1/(R_\theta^{3/2}\sqrt{d_+m})\le r_i,c_j\le
R_\theta^{3/2}/\sqrt{\theta d_-m}$.

Finally, put $W=Y^*(Y^*)^{\mathsf T}$ and
$a_\theta:=\min\{1/2,\eta\theta^2/(R_\theta^2d_+^2)\}$.
The codegree hypothesis and \cref{eq:bipartite-pseudorandom-diffuse}
give $W_{ii'}\ge a_\theta/m$ for distinct rows.  The same quadratic-form
argument as in the fixed-margin proof yields
$\sigma_2(Y^*)^2\le1-a_\theta$ and hence
$\sigma_2(Y^*)\le1-a_\theta/2$.
The uniform factor bounds and this singular-value gap verify
\cref{thm:bipartite-scaling-criterion}.

For the random-matrix assertion, set $\theta=1$, choose
$\rho=p/8$, and take $\eta=p^2/2$.  Then
$\eta>d_+-d_-$.  Chernoff bounds give the required row degrees, column
degrees, and codegrees with probability $1-o(1)$.  To verify total
support, condition on a fixed entry $B_{ij}=1$ and delete row $i$ and
column $j$.  The remaining matrix is an independent
$\operatorname{Bernoulli}(p)$ matrix of order $s=m-1$.  The Hall estimate
\cref{eq:random-bipartite-hall} gives failure probability
$e^{-\Omega_p(m)}$.
A union bound over the $m^2$ possible entries shows that, with probability
$1-o(1)$, every positive entry belongs to a supported permutation.
\end{proof}

\section{Auxiliary analytic estimates}
\label[appendix]{app:analytic-estimates}

This appendix contains the two scalar-factor calculations, the
analytic-continuation argument, Gamma-tail localization, cutoff calculus,
and matching resolvent estimates deferred from the main text.

\subsection{The matching scalar factor}
\label{subsec:matching-factor-proof}

\begin{proof}[Proof of \cref{lem:matching-factor}]
Write $z=a+ib$.  The inequality $\log u\le u-1$, applied to
$u=(1+a)^2+b^2$, gives
$\log|1+z|\le a+(a^2+b^2)/2$.  Hence
$\log|g(z)|\le a^2$, with the value at $z=-1$ understood by continuity.
When $|z|\le1/2$,
\[
  \log g(z)=\log(1+z)-z+\frac{z^2}{2}
  =\sum_{j\ge3}\frac{(-1)^{j+1}}jz^j
\]
has modulus at most $2|z|^3$.  It follows that
$|g(z)-1|\le2e^{1/4}|z|^3\le3|z|^3$.
\end{proof}

\subsection{The Ising scalar factor}
\label{subsec:ising-factor-proof}

\begin{proof}[Proof of \cref{lem:ising-factor}]
Write $z=x+iy$ with $x,y\in\R$.  The identity and envelope estimate are
\begin{align*}
  |\cosh(x+iy)|^2
  &=\sinh^2x+\cos^2y\le\cosh^2x,\\
  |g_{\mathrm{Is}}(x+iy)|
  &\le e^{-(x^2-y^2)/2}\cosh x
  \le e^{y^2/2}.
\end{align*}
Here the second line uses $\log\cosh x\le x^2/2$.
For the local bound, put $r=|z|$ and
$d_k=(2^k k!)^{-1}-(2k)!^{-1}>0$ for $k\ge2$.  Since $d_2=1/12$ and
$d_k\le[12\,2^{k-2}(k-2)!]^{-1}$ for $k\ge3$, comparison of the two power
series gives
\begin{align*}
  |\cosh z-e^{z^2/2}|
  &\le\frac{r^4}{12}e^{r^2/2},\\
  |g_{\mathrm{Is}}(z)-1|
  &\le\frac{r^4}{12}e^{r^2}
  \le\frac{r^4}{9}\le\frac{r^4}{8}
  \qquad(r\le1/2),
\end{align*}
where $e^{1/4}\le4/3$.  The same series begins with $-z^4/12$, which proves
\cref{eq:ising-quartic-zero} and the stated local constants.
\end{proof}

\subsection{Complex Gaussian marginalization}
\label{subsec:gaussian-marginalization-proof}

\begin{proof}[Proof of \cref{lem:gaussian-marginalization}]
The assertion is immediate for $S=\varnothing$, so assume that $S$ is
nonempty.  In the full representation, write
$A_Kx=u+iv$, where $u=K_+^{1/2}x$ and $v=K_-^{1/2}x$.  Let
$\zeta=\sigma+i\tau$.  For every $\eta>0$, the inequality
$2|rs|\le\eta r^2+\eta^{-1}s^2$ gives
\begin{align*}
 &\sum_{i\in S}
   \left(a_{\mathrm R}(\operatorname{Re}(\zeta\phi_i))^2
        +a_{\mathrm I}(\operatorname{Im}(\zeta\phi_i))^2\right)\\
 &\qquad\le
 (1+\eta)\sigma^2x^{\mathsf T}K_{\mathrm{env}}x
 +(1+\eta^{-1})\tau^2x^{\mathsf T}K'x,
\end{align*}
where $K'=a_{\mathrm R}K_-+a_{\mathrm I}K_+\preceq a_{\max}|K|$.
Indeed,
$\operatorname{Re}(\zeta\phi_i)=\sigma u_i-\tau v_i$ and
$\operatorname{Im}(\zeta\phi_i)=\sigma v_i+\tau u_i$, and extending the
sums from $S$ to all coordinates can only increase the two nonnegative
quadratic forms.

Choose $\eta,\eta'>0$ so that $(1+\eta)(1+\eta')^2(1-d)<1$.
Since $\|K\|\le\|K\|_{\mathrm{row}}\le\beta$, choose $\omega>0$ so small
and define
\[
  \begin{gathered}
  q:=(1+\eta)(1+\eta')^2(1-d)
  +2(1+\eta^{-1})\omega^2a_{\max}\beta<1,\\[2pt]
  \Omega:=\{\zeta\in\C:
  |\operatorname{Re}\zeta|<1+\eta',\ 
  |\operatorname{Im}\zeta|<\omega\}.
  \end{gathered}
\]
This is a connected open neighborhood of $[0,1]$.  By
\eqref{eq:diffuse-covariance-assumptions},
$2K_{\mathrm{env}}\preceq(1-d)I$, while
$K'\preceq a_{\max}\beta I$.  Uniformly on the closure of every slightly
smaller rectangle inside $\Omega$,
\[
  |F(\zeta(\phi_i)_{i\in S})|
  \le A\exp\left(\frac q2\|x\|_2^2\right).
\]
After multiplication by the standard Gaussian density, the right-hand side
is integrable.  Thus
$\zeta\mapsto\E_KF(\zeta(\phi_i)_{i\in S})$ is holomorphic on $\Omega$.

For the restricted representation put $L=K[S]$.  Interlacing gives
$\lambda_{\max}(L)\le\lambda_{\max}(K)$ and
$\lambda_{\min}(L)\ge\lambda_{\min}(K)$.
The condition $2K_{\mathrm{env}}\preceq(1-d)I$ bounds
$2a_{\mathrm R}\lambda_{\max}(K)_+$ and
$2a_{\mathrm I}(-\lambda_{\min}(K))_+$ by $1-d$.  Hence
$2(a_{\mathrm R}L_++a_{\mathrm I}L_-)\preceq(1-d)I$.
Also $\|L\|\le\|K\|\le\beta$, so the restricted expectation is
holomorphic on the same neighborhood $\Omega$.

The Taylor coefficient of order $r$ at zero of either holomorphic function
is a finite linear combination of moments of homogeneous polynomials of
degree $r$.  Wick's rule expresses each such moment only through the
bilinear second moments, which are $K[S]$ in both representations.  The two
Taylor series therefore agree.  The identity theorem on $\Omega$ gives
equality at $\zeta=1$, and the same majorants prove absolute convergence.
\end{proof}

\subsection{Localization of the Gamma average}
\label{subsec:gamma-localization-proof}

\begin{proof}[Proof of \cref{prop:radial-localization}]
Use the notation of \cref{sec:matching-radial}.  Put
$\rho_E=\|E\|$.  The entrywise hypothesis gives
$\|E\|_{\mathrm F}\le\beta$ and hence
$\tr|E|\le\beta\sqrt n$.  For every $r>\rho_E$,
\[
  \E_x\exp\left(\frac{\mathcal E_E(x)}{2r}\right)
  =\det(I-|E|/r)^{-1/2}
  \le\exp\left\{\frac{\tr|E|}{2(r-\rho_E)}\right\}.
\]
In particular, this is $e^{O(\sqrt n)}$ at $r=r_0$ and $r=r_1$,
uniformly over the matrices in the proposition.

If $U\ge r_1$, then \cref{eq:radial-amgm} and $1+z\le e^z$ give
$|Z_E(U^{-1})|\le\E_xe^{\mathcal E_E(x)/(2r_1)}$.  The Gamma Chernoff
bound therefore gives
\[
  \E\left[|Z_E(U^{-1})|\ind{U\ge r_1}\right]
  \le\exp\left\{-\frac n2(r_1-1-\log r_1)+O(\sqrt n)\right\}.
\]

For $0<u\le r_0$, the tangent-line inequality for the logarithm gives
\[
  \left(u+\frac{\mathcal E_E(x)}{n}\right)^{n/2}
  \le r_0^{n/2}
  \exp\left\{
    \frac{n(u-r_0)}{2r_0}+\frac{\mathcal E_E(x)}{2r_0}
  \right\}.
\]
Since $a_n-1-n/2=(\tau-1)/2$, integrating
\cref{eq:radial-amgm} against \cref{eq:gamma-density} bounds the lower-tail
contribution by
\[
  \frac{b_n^{a_n}}{\Gamma(a_n)}r_0^{n/2}
  \E_xe^{\mathcal E_E(x)/(2r_0)}
  \int_0^{r_0}u^{(\tau-1)/2}
  \exp\left\{-b_nu+\frac{n(u-r_0)}{2r_0}\right\}\,du.
\]
Write the exponent in the integral as $-n/2+\lambda_nu$, where
$\lambda_n=n/(2r_0)-b_n=\Omega(n)$.
It is increasing for all sufficiently large $n$.  On $[r_0/2,r_0]$, the
factor $u^{(\tau-1)/2}$ is bounded, and endpoint integration gives
\[
  \int_{r_0/2}^{r_0}u^{(\tau-1)/2}e^{-n/2+\lambda_nu}\,du
  \le\frac{C}{\lambda_n}e^{-n/2+\lambda_nr_0}
  =\frac{C}{\lambda_n}e^{-b_nr_0}.
\]
On $(0,r_0/2]$, the exponential is at most
$e^{-n/2+\lambda_nr_0/2}=e^{-b_nr_0-\lambda_nr_0/2}$.
The factor $u^{(\tau-1)/2}$ is integrable at zero for both
$\tau\in\{0,1\}$, while $\lambda_nr_0/2=\Omega(n)$.  This part is
exponentially smaller than the upper-half contribution.  The integral is
therefore at most a fixed multiple of $e^{-b_nr_0}$.  Stirling's formula
now gives
\[
  \E\left[|Z_E(U^{-1})|\ind{U\le r_0}\right]
  \le
  \exp\left\{-\frac n2(r_0-1-\log r_0)
  +O(\sqrt n+\log n)\right\}.
\]
Both rate functions are positive.  Decreasing $c$ and increasing $n_0$
proves \cref{eq:radial-tail}.
\end{proof}

\subsection{The least truncation cutoff}
\label{subsec:cutoff-proof}

\begin{proof}[Proof of \cref{lem:generic-cutoff}]
Put $D=C_2^{1/\alpha}$ and $A_\alpha=2+2/\alpha$.  Choose $\eta>0$
so small that $A_\alpha\eta\le c_{\mathrm{tr}}$ and
$DA_\alpha\eta\le1/4$.
The inequality \eqref{eq:least-cutoff-condition} is equivalent to
$\left(DR/n\right)^R\le2^{-b/\alpha}$.  Consider
\[
  R_*=\max\left\{2,\left\lceil\frac{2b}{\alpha}\right\rceil\right\}.
\]
Since $b\ge1$ and $b<\eta n$, we have
$R_*\le A_\alpha b<A_\alpha\eta n\le c_{\mathrm{tr}}n$.  Moreover,
$DR_*/n\le1/4$ and $R_*\ge2b/\alpha$, so
$(DR_*/n)^{R_*}\le4^{-R_*}\le2^{-b/\alpha}$.
Thus a least admissible $R$ exists and is at most $R_*$.

For the enumeration bound, define $F(x)=x\log(n/(Dx))$.
On $2\le x\le R_*$, we have $Dx/n\le1/4$, so
$F'(x)=\log(n/(Dx))-1\ge\log4-1>0$.  If $R>2$, minimality gives
$\alpha F(R-1)<b\log2$.  The mean-value theorem and the fixed value of $D$
give $F(R)-F(R-1)=O(\log n)$, and hence
$F(R)=O(b+\log n)$.  The same conclusion is immediate when $R=2$.
Also $F(R)\ge R\log4$, so $R=O(b+\log n)$.  Finally,
$R\log(en/R)=F(R)+R\log(eD)=O(b+\log n)$,
which proves \eqref{eq:least-cutoff-cost}.
\end{proof}

\subsection{Matching resolvent estimates}
\label{subsec:matching-resolvent-estimates}

\begin{proof}[Proof of \cref{lem:uniform-completion}]
We prove the spectral, entrywise, and determinant bounds in that order.
Eigenvalue calculus handles the first, the resolvent expansion the second,
and the identity $\tr E=0$ the third.

If $e$ is an eigenvalue of $E$, the corresponding eigenvalue of $K_t$ is
$te/(1+te)$.  This function is increasing, and $|te|\le q_0$.  We may
therefore take
\(d=(1-q_0)/(1+q_0)\).

Put $\rho=\|E\|$, and let $\mathbf e_i$ denote the $i$th coordinate
vector.  For $\ell\ge2$, the entrywise hypothesis on $E$ gives
\[
  |(E^\ell)_{ij}|
  \le\|E\mathbf e_i\|_2\|E\|^{\ell-2}\|E\mathbf e_j\|_2
  \le\frac{\beta^2}{n}\rho^{\ell-2}.
\]
Expanding $K_t$ as a convergent resolvent series proves the entrywise bound
in \cref{eq:uniform-K}.  For example, one may take
$\beta'=\beta/r_0+(\beta^2/r_0^2)/(1-q_0)$.

Finally, $\tr E=0$ because $E$ has zero diagonal.  If
$\lambda_1,\ldots,\lambda_n$
are its eigenvalues, then
$\log\det(I+tE)=\sum_i(\log(1+t\lambda_i)-t\lambda_i)$.
The summands have absolute value at most
$t^2\lambda_i^2/(2(1-q_0))$.  Since $\|E\|_{\mathrm F}\le\beta$, their sum is
bounded by a constant.
\end{proof}

\section{Bit complexity}
\label[appendix]{app:bit-complexity}

The main text describes the algorithms in real arithmetic.  This appendix
records the stability and precision estimates that give the stated bit
complexities.

\subsection{Perfect matchings}

Retain the notation $b,R,C,c_*,K_t$, and $Z_{E,R}$ from
\cref{sec:matching-algorithm}.  For a real symmetric matrix $F$ and a scalar
$u$, write
$K_{u,F}=uF(I+uF)^{-1}$ and abbreviate $Z_R(t,E)=Z_{E,R}(t)$.

\begin{lemma}[Stability of the truncated formula]
\label{lem:numerical-stability}
Fix $q_{\mathrm{stab}}<1$ and $C_F<\infty$.  There is a constant
$C_{\mathrm{stab}}$ with the
following property.  Let $2\le R\le c_*n$, and suppose that every pair
$(u,F)$ on the line segment joining $(t,E)$ and
$(\widetilde t,\widetilde E)$ satisfies
\[
  \operatorname{tr}F=0,
  \qquad \|F\|_{\mathrm F}\le C_F,
  \qquad \|uF\|\le q_{\mathrm{stab}},
  \qquad
  \|K_{u,F}[T]\|_{\mathrm{row}}\le\frac14
  \quad(|T|<R).
\]
Then
\[
  |Z_R(\widetilde t,\widetilde E)-Z_R(t,E)|
  \le n^{C_{\mathrm{stab}}}
  \exp\left(C_{\mathrm{stab}}R\log\frac{en}{R}\right)
  \left(|\widetilde t-t|+
  \|\widetilde E-E\|_{\max}\right).
\]
The constant $C_{\mathrm{stab}}$ depends only on the fixed margins.
\end{lemma}

\begin{proof}
The differential of the transformed matrix is
$dK=(I+tE)^{-1}(dt\,E+t\,dE)(I+tE)^{-1}$.
The resolvent margin and
$\|dE\|\le n\|dE\|_{\max}$ contribute only a polynomial factor.
The trace, Frobenius, and spectral assumptions likewise control the
determinant factor and its derivative.  For a fixed set $T$, the row-norm
condition gives
$\|(I-K_{u,F}[T])^{-1}\|_{\mathrm{row}}\le4/3$.
Differentiating \cref{eq:fixed-support-formula} and its monomer--dimer
recurrence therefore bounds its value and derivative by
$n^{O(1)}e^{O(|T|)}$.  Inclusion--exclusion over $T\subseteq S$
contributes a factor $2^{|S|}$.  Summing over $|S|<R$ and applying the mean
value theorem proves the claim.
\end{proof}

\begin{remark}[Numerical precision for the matching algorithms]
Assume that the input and $\eps$ are rational, and let $L_{\mathrm{in}}$
denote their total encoding length.  Writing
$r_v=n^{-1/2}e^{z_v}$, the symmetric scaling is the minimizer of
\[
  \Psi_A(z)=\frac1n\sum_{uv\in E(G)}a_{uv}e^{z_u+z_v}-\sum_v z_v.
\]
The bipartite scaling is obtained similarly by minimizing
\[
  \Psi_B(x,y)=\frac1m\sum_{b_{ij}>0}b_{ij}e^{x_i+y_j}
  -\sum_i x_i-\sum_j y_j
\]
on the gauge subspace $\sum_i x_i=\sum_j y_j$.  The factor bounds and
spectral gaps in \cref{sec:matching-normalization} make these potentials
uniformly strongly convex and smooth on fixed boxes containing their
minimizers.  Standard projected gradient descent therefore finds the
scaling to accuracy $2^{-\nu}$ in $O(\nu+\log n)$ iterations
\cite[Chapter~2]{Nesterov2018}.

The entropy exponent is recovered from the same vector with error at most
$\sqrt n\,2^{-\nu}$, since
$h_A(G)=\frac n2\log n-\sum_vz_v$ by the proof of
\cref{prop:entropy-normalization}, and
$h_B=m\log m-\sum_ix_i-\sum_jy_j$ likewise.  Let $h$ denote the relevant
entropy and let $\widetilde h$ denote its computed approximation.  Choose
$\nu$ so that $|\widetilde h-h|\le2^{-b}$, which changes the final entropy
factor by at most $e^{\pm2^{-b}}$.

Use a fixed $\bar q_0\in(q_0,1)$ as $q_{\mathrm{stab}}$ in the stability
lemma.  The exact bounds
$\|tE\|\le q_0$ and $\|K_t[T]\|_{\mathrm{row}}\le1/8$, together with
continuity, imply that every sufficiently accurate approximation and the
line segment joining it to the exact pair satisfy
$\|uF\|\le\bar q_0$ and
$\|K_{u,F}[T]\|_{\mathrm{row}}\le1/4$.

By \cref{lem:numerical-stability,eq:R-runtime}, the retained formula has
size $n^{O(1)}2^{O(b)}$ and is Lipschitz with a constant of the same form.
Its inverses remain uniformly conditioned because
$I+tE\succeq(1-q_0)I$ and $I-K_t[T]\succeq3I/4$.  A gate-by-gate
absolute-error induction bounds the accumulated rounding radius by the
circuit size times $\exp(O(n\log n))$ times the unit roundoff.  No division
is made by a signed subset contribution.  Thus outward-rounded dyadic
evaluation with $O(n\log(n+2)+b+L_{\mathrm{in}})$ working bits makes the
scaling, input-perturbation, and rounding errors each at most $2^{-b}$.
Together with the three analytic errors, these give an additive error in
$Q_E$ of at most $6\cdot2^{-b}\le6\eps/C^2$, hence a relative error of at
most $\eps/2$.  The entropy error was budgeted separately above.

In the exact branch, the subset recurrence is applied directly to the
original matrix $A$, or to the symmetric dilation of $B$, rather than to the
generally irrational scaling.  Clearing denominators once gives intermediate
integers of polynomial bit length, and the recurrence returns zero when the
hafnian is zero.  This proves the bit-complexity bounds in
\cref{thm:main,thm:scaling-criterion,thm:bipartite-scaling-criterion}.
\end{remark}

\subsection{The Ising model}

Retain $b,R$, and $K$ from the proof of \cref{thm:ising}.  The number of
elementary exponential terms is
\begin{equation}
  M_R=\sum_{s<R}\binom ns3^s,
  \qquad
  \log M_R=O_{\beta,\kappa}(b+\log n),
  \label{eq:ising-term-count}
\end{equation}
by \eqref{eq:ising-cutoff-size}.  The row-norm clause of
\cref{thm:gaussian-product-truncation} gives
$\|K[T]\|_{\mathrm{row}}\le1/8$ for every retained set.  Hence
\cref{lem:ising-retained-term} applies, its exponential arguments have
magnitude $O(R)$, and $I+K[T]\succeq7I/8$.  The rational matrices
$K$, $K[T](I+K[T])^{-1}$, and the relevant determinants have bit length
$\operatorname{poly}(n,L_J)$.

Evaluate each elementary term to absolute error at most
$2^{-b-2}/M_R$ and the final prefactor to relative error $2^{-b-2}$.  This
requires $\operatorname{poly}(n,L_J)+O_{\beta,\kappa}(b+\log n)$ working
bits.  Standard algorithms for rational linear algebra, square roots, and
exponentials then take polynomial time per term.  Together with
\eqref{eq:ising-term-count}, this proves the bit-complexity bound in
\cref{thm:ising}.

\end{document}